%% file: IEEE-Published.tex
\documentclass[journal,twocoulumn]{IEEEtran}

\input{IEEE-preambole}

\input{def}

\usepackage{diagbox}
\newcommand{\M}{\pazocal{U}}

\begin{document}


\title{Completing the Grand Tour of asymptotic quantum coherence manipulation}


\author{Ludovico~Lami %
\thanks{Ludovico Lami is with the School of Mathematical Sciences and Centre for the Mathematics and Theoretical Physics of Quantum Non-Equilibrium Systems, University of Nottingham, University Park, Nottingham NG7 2RD, United Kingdom. Email: ludovico.lami@gmail.com}}



\maketitle

\begin{abstract}
We compute on all quantum states several measures that characterise asymptotic quantum coherence manipulation under restricted classes of operations. We focus on the distillable coherence, i.e.\ the maximum rate of production of approximate pure bits of coherence starting from independent copies of an input state $\boldsymbol{\rho}$, and on the coherence cost, i.e.\ the minimum rate of consumption of pure coherence bits that is needed to generate many copies of $\boldsymbol{\rho}$ with vanishing error. We obtain the first closed-form expression for the distillable coherence under strictly incoherent operations (SIO), proving that it coincides with that obtained via physically incoherent operations (PIO). This shows that SIO and PIO are equally weak at distilling coherence, sheds light on the recently discovered phenomenon of generic bound coherence, and provides us with an explicit optimal distillation protocol that is amenable to practical implementations. We give a single-letter formula for the coherence cost under PIO, showing that it is finite on a set of states with nonzero volume. Since PIO can be realised in a laboratory with incoherent ancillae, unitaries, and measurements, our result puts fundamental limitations on coherence manipulation in an experimentally relevant setting. We find examples of `abyssally bound' states with vanishing PIO distillable coherence yet infinite PIO coherence cost. Our findings complete the picture of asymptotic coherence manipulation under all the main classes of incoherent operations.
\end{abstract}

%
\IEEEpeerreviewmaketitle

\section{Introduction}

\subsection{Resource theory of quantum coherence }

Coherent superposition of states can be regarded as the fundamental quantum feature from which all the other wonders of quantum theory, such as entanglement and in turn nonlocality, descend. In spite of its central role for theory and applications, a rigorous framework to study the manipulation of quantum coherence -- what is called a resource theory~\cite{Brandao-Gour, RT-review} -- has been identified only recently~\cite{Aberg2006, Baumgraz2014, Levi2014, Winter2016, coherence-review}. Two main ingredients are needed in order to define a resource theory: free states and free operations. Once these objects have been identified, questions of information-theoretical nature arise naturally: how efficiently can we convert a given state into standard units of resource by using free operations? Conversely, how many units of resource do we need to invest to prepare a given target state with free operations? In the history of quantum information theory, operational questions of this sort have been asked first in the context of entanglement theory~\cite{Bennett-distillation, Horodecki-review}, from which the terminology is borrowed: the first task is traditionally referred to as \emph{distillation}, the second as \emph{formation} or \emph{dilution}. Following the glorious tradition initiated by Shannon~\cite{Shannon}, we will look at the asymptotic regime only, with the motivation that it captures ultimate limitations on experimental capabilities.

In coherence theory, free states are represented by density matrices that are diagonal in a fixed orthonormal basis $\{\ket{i}\}_i$. To fix ideas, we can think of the index $i$ as labelling the different arms of an interferometer where a single photon has been injected~\cite{Biswas2017}. Free states correspond to the photon being probabilistically localised in a single arm. In spite of this simplicity at the level of states, identifying the `correct' set of free operations has instead been subject of debate. Since free operations must preserve the set of free states, one can define \emph{maximal incoherent operations} (MIO) as those quantum channels $\Lambda$ such that $\Lambda(\delta)$ is diagonal for any diagonal density matrix $\delta$~\cite{Aberg2004, Aberg2006}. However, quantum channels can be represented in Kraus form as $\Lambda(\cdot)= \sum_\alpha K_\alpha (\cdot) K_\alpha^\dag$, and in the physical interpretation of said channel as an instrument each $\alpha$ corresponds to a different measurement outcome. In light of this, we may instead demand that each Kraus operator be incoherent, i.e.\ such that $K_\alpha \ket{i} \propto \ket{j_{\alpha, i}}$, where we stressed that the output label $j$ can depend on both $i$ and $\alpha$. \emph{Incoherent operations} (IO) are those that admit a Kraus representation with this property~\cite{Baumgraz2014}. To define \emph{strictly incoherent operations} (SIO) one requires instead that both $K_\alpha$ and $K_\alpha^\dag$ be incoherent. It is easy to verify that every SIO $\Lambda$ commutes with the dephasing operator $\Delta$ that erases all off-diagonal elements, in formula $[\Lambda, \Delta]=0$ as superoperators. When taken on its own, this latter identity defines the class of \emph{dephasing-covariant incoherent operations} (DIO)~\cite{Chitambar-PIO, Chitambar2016, Marvian2016}. While all SIO are also DIO, these form a strictly larger set which is also different from that of IO. Remarkably, SIO admit an operational description as concatenations of elementary processes that involve the system plus some ancillae and are incoherent on the former while possibly coherent on the latter~\cite{Yadin2016}. The more restricted paradigm of \emph{physically incoherent operations} (PIO) requires instead that all such operations be globally incoherent~\cite{Chitambar-PIO, Chitambar2016}. Other classes of operations that have been defined in the literature~\cite{Gour2008, deVicente2016} will not be considered further here.

The foundations of an operational theory of coherence were laid by Winter and Yang~\cite{Winter2016}, who showed that the MIO/DIO/IO distillable coherence is given by the \emph{relative entropy of coherence}~\cite{Aberg2006, Baumgraz2014} for all states, and also proved that the IO/SIO coherence cost coincides with the \emph{coherence of formation}~\cite{Aberg2006}. Previously, the problems of distillation and formation had been fully resolved for pure states only~\cite{Xiao2015}, in which case all the above measures reduce to the \emph{entropy of coherence}. These results imply that while the resource theory of coherence is not reversible under IO, these operations leave no bound coherence, i.e.\ some coherence can always be distilled from states with nonzero cost. This is in stark contrast with the existence of bound states in other resource theories, most notably that of entanglement under either local operations and classical communication (LOCC)~\cite{HorodeckiBound}. Later, Chitambar~\cite{Chitambar-reversible} showed that coherence becomes a reversible resource theory when the free operations are taken to be MIO and even DIO. This significant strengthening of the general result in~\cite{Brandao-Gour} entails that the same relative entropy of coherence characterises the MIO/DIO coherence cost.

As it appears from Table~\ref{table operations/rates}, this wealth of results leaves however three main questions open, namely the evaluation of the SIO and PIO distillable coherence, and that of the PIO coherence cost. For pure states, it is known that the first two quantities coincide once again with the entropy of coherence~\cite{Xiao2015, Chitambar-PIO, Chitambar2016}, while in~\cite{Chitambar-PIO, Chitambar2016} it is also shown that all pure states that are not maximally coherent cannot be prepared via PIO starting from coherence bits, i.e.\ their PIO coherence cost is infinite. 
Notable progress on the problem of asymptotic SIO distillation of mixed states was presented in~\cite{Zhao2018}, where an SIO bound coherent state was constructed for the first time. Such state has provably zero SIO distillable coherence yet nonzero cost. While this is arguably surprising, what is even more surprising is that, unlike in entanglement theory, SIO bound entanglement is a \emph{generic} phenomenon, meaning that almost all states, in a measure-theoretic sense, are SIO (and hence PIO) bound coherent~\cite{bound-coherence}.
A full understanding of SIO/PIO distillable coherence and of PIO coherence cost has however remained out of reach so far.

\subsection{Our contributions}

In this paper we tackle and solve the three aforementioned problems (Table~\ref{table operations/rates}). First, we give an easily computable, analytical formula for the SIO/PIO distillable coherence that we dub \emph{quintessential coherence}. This substantiates the claims of~\cite{bound-coherence}, solves the conjecture that was proposed in the first version of that manuscript, and answers a question raised already in~\cite{Winter2016, Chitambar2016, Yadin2016}.
Our result show that SIO and PIO, while behaving radically differently at the single-copy level~\cite{Chitambar2016}, possess the same distillation power in the asymptotic regime. This is especially remarkable given that several preliminary results seemed to rather indicate a substantial equivalence between SIO and IO. For instance, SIO are as powerful as IO in pure-to-pure state transformations~\cite{Du2015, Torun2019} and in (asymptotic) coherence dilution~\cite{Winter2016}; they perform no worse than DIO at probabilistic distillation from pure states~\cite{Fang2018}. Sometimes even the largest set of MIO does not give any advantage over SIO: this is the case e.g.\ for one-shot distillation from pure state~\cite{Regula2018} and for assisted distillation~\cite{bartosz-myself-alex}. All this evidence has indeed led some authors to speculate that the SIO and IO distillable coherence may be equal~\cite{Winter2016}, while some others advocated the importance and centrality of SIO in light of their strong operational interpretation~\cite{Yadin2016}. In this work we settle all these questions by quantifying exactly the SIO distillable coherence on all states, and showing that -- somewhat unfortunately -- it coincides with that obtainable with the much more restricted set of PIO.
While the direct part of our statement is proved in a relatively intuitive way, the main technical challenge lies in establishing the converse. To achieve this, we introduce a whole family of coherence monotones tailored to SIO, connect them to the smooth conditional max-entropies of some \emph{classical} random variables derived from the underlying quantum state, and conclude by applying a tweaked version of the asymptotic equipartition property.

Our second contribution is a single-letter formula for the PIO coherence cost that resembles the coherence of formation, except that the only allowed states in the convex decomposition are uniformly coherent on a subset of indices: we dub such a quantity the \emph{uniform coherence of formation}. Our findings demonstrate that the set of states with finite PIO coherence cost has nonzero volume, as it contains a whole ball centred around the maximally mixed state. This somehow counterintuitive result is in spite of the fact that all pure states but the maximally coherent ones have infinite PIO coherence cost~\cite{Chitambar-PIO}. Again, the crux of the argument is the proof of the converse. In fact, we are unable to decide whether the uniform coherence of formation obeys some form of \emph{asymptotic continuity}, which is a notoriously instrumental property for establishing a converse statement. We circumvent this difficulty by proving and exploiting its superadditivity and lower semicontinuity instead; to the best of our knowledge, this proof strategy is relatively original and may be of independent interest.

\begin{table}[h] \centering
\ra{1.5}
\begin{tabular}{lcc} \toprule
Operations & Distillable coherence & Coherence cost \\
\midrule
MIO & $C_r$ & $C_r$ \\
DIO & $C_r$ & $C_r$ \\
IO & $C_r$ & $C_f$ \\
SIO & \cellcolor{gray!15} $Q$ & $C_f$ \\
PIO & \cellcolor{gray!15} $Q$ & \cellcolor{gray!15} $C_f^\M$ \\
\bottomrule
\end{tabular}
\vspace{1.5ex}
\caption{The distillable coherence and coherence cost under the main classes of free operations. 
See Eq.~\eqref{relative entropy of coherence},~\eqref{coherence of formation},~\eqref{Q thm}, and~\eqref{PIO cost} for definitions. The MIO/DIO/IO distillable coherence and the IO/SIO coherence cost were computed in~\cite{Winter2016} (see also~\cite{Regula2018, Chitambar-reversible}), while the MIO/DIO coherence cost was first determined in~\cite{Chitambar-reversible}. 
Our contribution is to fill in the last three entries of this table, highlighted in grey. See Theorems~\ref{solution thm} and~\ref{C c PIO thm}, proven in Sections~\ref{sec converse} and~\ref{sec PIO cost}.} 
\label{table operations/rates}
\end{table}

The rest of the paper is structured as follows. In Section~\ref{sec main results} we introduce the reader to the theory of quantum coherence and state our main results formally. Section~\ref{sec direct} describes the SIO distillation protocol found in~\cite{bound-coherence}. In Section~\ref{sec monotones} we introduce a new family of SIO monotones, which are subsequently used in Section~\ref{sec converse} to establish the optimality of said distillation protocol, thus calculating the SIO distillable coherence on all states. In the subsequent Section~\ref{sec PIO cost} we tackle and solve the question of computing the PIO coherence cost. Finally, in Section~\ref{sec conclusions} we discuss our results, draw our conclusions and present some open problems and directions for future research.

\section{Main results} \label{sec main results}

\subsection{Basic definitions}

We consider a $d$-dimensional quantum system whose Hilbert space $\C^d$ is spanned by some preferred basis $\{\ket{i}\}_{i=1,\ldots, d}$, referred to as the \emph{incoherent} (or \emph{computational}) \emph{basis}. As already mentioned, we can imagine that the state $\ket{i}$ represents a single photon localised in the $i$-th arm of an interferometer.
A generic \emph{incoherent state} will be represented as a diagonal density matrix $\delta =\sum_{i=1}^d \delta_i \ketbra{i}$, where the coefficients $\delta_i\geq 0$ form a probability distribution. A pure state of the form
\bb
\ket{\Psi} = \frac{1}{\sqrt{k}} \sum_{j\in J} e^{i \theta_j} \ket{j}\, ,
\label{max coh}
\ee
where $J\subseteq [d]\coloneqq \{1,\ldots, d\}$ is a subset of indices of cardinality $|J|=k$ and $\theta_j\in\R$ are phases, is called a \emph{uniformly coherent state} of size $k$. Such states are particular examples of the more general class of $k$-coherent states studied in~\cite{Ringbauer2017}. Customarily, a uniformly coherent state of size $d$ is called a \emph{maximally coherent state}. In what follows, we will denote by $\M_k$ the sets of uniformly coherent states of size $k$, and we will also set
\bb
\M \coloneqq \bigcup_{k=1}^d \M_k\, ,
\label{M}
\ee
pure states in $\M$ being generically referred to as uniformly coherent.

A maximally coherent state of a single qubit is usually called a \emph{coherence bit}. For the canonical choice of all phases equal to zero, we can write it as $\ket{\Psi_2} \coloneqq \frac{1}{\sqrt2}\left( \ket{0}+\ket{1} \right)$. In what follows, for a generic pure state $\ket{\phi}$ we will denote the corresponding density matrix as $\phi \coloneqq \ketbra{\phi}$.

A $d\times d'$ matrix $K$ is said to be an \emph{incoherent operator} if $K\ket{i}$ is proportional to some vector of the incoherent basis on $\C^{d'}$, i.e.\ $K\ket{i} = c_{i} \ket{f(i)}$ for some function $f:[d]\to [d']$ and some complex-valued function $c:[d]\to \C$. With this definition, it is elementary to show that \emph{incoherent unitaries} are those that can be written as a product of a permutation and a diagonal matrix, in formula $U=\sum_j e^{i\theta_j} \ketbraa{\pi(j)}{j}$. Similarly, \emph{incoherent projectors} are of the form $\Pi_I \coloneqq \sum_{i\in I}\ketbra{i}$ for some subset of indices $I\subseteq [d]$. We remind the reader that a generic quantum channel\footnote{A linear map $\Phi:M_n(\C) \to M_m(\C)$ between the algebras of complex matrices of size $n$ (input) and $m$ (output) is called a \emph{quantum channel} if it is: (i)~completely positive, meaning that for all integers $k$ and all positive semidefinite matrices $A \geq 0$ of size $nk$ it holds that $(\Phi\otimes \Id_k)(A)\geq 0$; and (ii)~trace-preserving, i.e.\ such that $\Tr \Phi(X) = \Tr X$ for all $X\in M_n(\C)$.} can be written in \emph{Kraus form} as
\bb
\Lambda(\cdot) = \sum_\alpha K_\alpha (\cdot) K_\alpha^\dag\, ,
\label{Kraus}
\ee
where we assume without loss of generality that the range of the index $\alpha$ is finite. Here we will mainly care about two classes of incoherent operations, that we set out to define now. 
After stating the mathematically precise definitions, we will elaborate on the physical and operational motivations behind it.
As explained in~\cite[SM,~Eq.~(S9)]{bound-coherence}, up to `lifting and compressing' input and output by means of incoherent isometries and incoherent projectors, we can without loss of generality look at the case where input and output dimension coincide.

\begin{Def}[\cite{Winter2016}]
A \emph{strictly incoherent operation} (SIO) is a quantum channel that admits a Kraus representation as in Eq.~\eqref{Kraus}, where for each $\alpha$ both $K_\alpha$ and $K_\alpha^\dag$ are incoherent.
\end{Def}

\begin{rem}
In other words, an SIO acts as
\bb
\Lambda(\cdot) = \sum\nolimits_\alpha U_{\pi_\alpha} D_\alpha (\cdot) D_\alpha^* U_{\pi_\alpha}^\intercal\, ,
\label{SIO Kraus}
\ee
where the $\pi_\alpha\in S_d$ are permutations, $U_{\pi_\alpha}\coloneqq \sum_{i=1}^d \ketbraa{\pi_\alpha(i)}{i}$ are the unitaries that implement them, and $D_\alpha\coloneqq \sum_{i=1}^d d_\alpha(i) \ketbra{i}$ are diagonal matrices. Observe that $\Lambda$ is trace-preserving iff $\sum_\alpha |d_\alpha(i)|^2 = 1$ for all $1\leq i\leq d$.
\end{rem}

\begin{Def}[\cite{Chitambar-PIO}]
A \emph{physically incoherent operation} (PIO) is a quantum channel $\Lambda$ that acts as
\bb
\Lambda(\cdot) = \sum_{\alpha, \beta} p_\alpha U_{\alpha,\beta} \Pi_{\beta|\alpha} (\cdot) \, \Pi_{\beta|\alpha} U_{\alpha,\beta}^\dag\, ,
\label{PIO Kraus}
\ee
where $p$ is an arbitrary probability distribution, $U_{\alpha,\beta}$ are incoherent unitaries, and for all $\alpha$ the $\Pi_{\beta|\alpha}$ are incoherent projectors forming a complete measurement, meaning that $\sum_\beta \Pi_{\beta|\alpha} = \id$ is the identity matrix. 
\end{Def}

\begin{rem} \label{elementary PIO rem}
In other words, a PIO can be thought of a convex combination of `elementary PIO' each acting as 
\bb
\Lambda(\cdot) = \sum_\beta U_\beta \Pi_\beta (\cdot) \Pi_\beta U_\beta^\dag\, ,
\label{elementary PIO Kraus}
\ee
where again the $U_\beta$ are incoherent unitaries and the $\Pi_\beta$ form a complete set of incoherent projectors.
\end{rem}

\begin{rem}
We should bear in mind that all PIO are SIO, but the converse need not be true~\cite{Chitambar-PIO}.
\end{rem}

SIO and PIO are meaningful because, unlike MIO, DIO, or IO, they admit explicit implementations in terms of clearly identifiable elementary operations. In the case of SIO, these elementary operations consist in~\cite{Yadin2016}: (i) appending ancillae prepared in incoherent states; (ii) performing unitary operations on the ancilla controlled by the incoherent basis of the system; (iii) making arbitrary destructive measurements on the ancilla; and (iv) applying any incoherent unitary on the system. Instead, to implement an arbitrary PIO no ancilla is necessary, and we only need~\cite{Chitambar-PIO, Chitambar2016} (iv) incoherent unitaries and (v) incoherent measurements, to wit, measurements represented\footnote{In the positive-operator valued measure (POVM) formalism, a quantum measurement with outcomes indexed by a label $i$ is represented by a (finite) collection of positive semidefinite operators $\{E_i\}_i$ that add up to the identity.} by incoherent projectors. It is always assumed that classical randomness is available for free.

To fix ideas, it may be helpful to review how some of the above operations may be implemented within our exemplary scheme involving a single photon in a multi-arm interferometer~\cite{Biswas2017}. If the wavelength of our photon is known with precision, incoherent unitaries in (iv) can be applied by permuting and changing the lengths of the different arms, e.g.\ by means of mirrors. Incoherent measurements described in (v) are also easy to implement, at least in principle: it suffices to deflect different groups of arms by means of mirrors; the classical outcome is retrieved by looking at which mirror has received a kick that is compatible with the deflection in the photon's momentum. To model (iii) one can think of a setting where the experimentalist can use a one-way quantum channel to send the ancilla to a distant facility where much more powerful equipment allows for arbitrary quantum operations to be performed. Since quantum communication is only one-way but classical communication is unrestricted, any classical output produced by the equipment can be disclosed to the experimentalist, allowing them to effectively implement arbitrary destructive measurements. While controlled unitaries as in (ii) may be regarded as the most experimentally challenging, the fact that they can be accessed in the SIO framework constitutes a precise and testable assumption. We stress once more that such assumptions do not seem to be identifiable in the MIO, DIO, or IO settings.

Several useful measures of coherence have been identified and studied so far. We limit ourselves to recalling the most significant in the context of asymptotic manipulation of coherence. The \emph{entropy of coherence} of a pure state $\phi=\ketbra{\phi}$ is simply defined as
\bb
C(\phi) \coloneqq S\left( \Delta(\phi)\right) ,
\label{entropy of coherence}
\ee
where $\Delta(\cdot)\coloneqq \sum_i \ketbra{i}(\cdot) \ketbra{i}$ is the dephasing map, and $S(\rho)\coloneqq -\Tr[\rho\log \rho]$ stands for the von Neumann entropy\footnote{Unless otherwise specified, in this paper logarithms are always assumed to be to base $2$.}. Such a measure can be extended to all mixed states via a convex roof construction. The resulting quantity is the \emph{coherence of formation}, defined as~\cite{Aberg2006}
\bb
C_f(\rho) \coloneqq \inf_{\sum_i p_i \psi_i=\rho} \sum_i p_i S\left( \Delta(\psi_i) \right) ,
\label{coherence of formation}
\ee
where the infimum is taken over all convex decompositions of $\rho$ into pure states $\psi_i$. Taking another viewpoint, one could try to quantify the coherence content of a quantum state by looking at its distance from the set of incoherent states. Using as metric the relative entropy $D(\rho\|\sigma)\coloneqq \Tr[\rho (\log \rho - \log \sigma)]$ yields the \emph{relative entropy of coherence}, given by~\cite{Aberg2006, Baumgraz2014}
\bb
C_r(\rho) \coloneqq \min_{\delta=\Delta(\delta)} D(\rho\|\delta) = S\left( \Delta(\rho)\right) - S(\rho)\, .
\label{relative entropy of coherence}
\ee
We will see in the next section how these measures play a role in characterising formation and distillation processes under some relatively large classes of incoherent operations. At the same time, we will learn how to construct alternative measures that capture the essential quantitative features of coherence manipulation under the smallest sets of free operations.

\subsection{SIO/PIO distillable coherence} \label{subsec SIO/PIO distillable}

The process of coherence distillation consists in extracting coherence bits $\Psi_2$ starting from a large supply of identical copies of a state $\rho$. Following the information-theoretical standard approach of looking at the asymptotic regime, one can define the \emph{distillable coherence} under a set of operations $\pazocal{O}$ as the maximal rate at which this process can be carried out with vanishing error: 
\bb
C_{d,\pazocal{O}}(\rho) \coloneqq \sup\left\{r: \lim_{n\to\infty} \inf_{\Lambda\in \pazocal{O}} \left\| \Lambda( \rho^{\otimes n} )\! -\! \Psi_{2^{\floor{rn}}} \right\|_1 = 0 \right\} .
\label{distillable}
\ee 
For IO/DIO/MIO it is known that~\cite{Winter2016, Regula2018, Chitambar-reversible}
\bb
C_{d,\IO}(\rho) = C_{d,\DIO}(\rho) = C_{d, \MIO}(\rho) = C_r(\rho)
\label{results VV}
\ee
for all states $\rho$, which gives an operational interpretation to the relative entropy of coherence defined in Eq.~\eqref{relative entropy of coherence}. Our first result allows us to evaluate the distillable coherence on the smaller classes SIO/PIO, solving a problem mentioned already in~\cite{Winter2016, Chitambar2016, Yadin2016} and filling in the two missing entries in the first column of Table~\ref{table operations/rates}.

\begin{thm} \label{solution thm}
For all states $\rho$, the distillable coherence under SIO/PIO satisfies
\bb
C_{d, \SIO}(\rho) = C_{d, \PIO}(\rho) = Q(\rho)\, ,
\label{solution}
\ee
where the \emph{quintessential coherence} is defined as
\bb
\begin{aligned}
Q(\rho) &\coloneqq S\left( \Delta(\rho) \right) - S\left( \widebar{\rho} \right) , \\[.5ex]
\widebar{\rho} &\coloneqq \sum_{(i,j):\, |\rho_{ij}| = \sqrt{\rho_{ii}\rho_{jj}}} \rho_{ij}\, \ketbraa{i}{j}\, .
\end{aligned}
\label{Q thm}
\ee
\end{thm}

The proof is presented in Section~\ref{sec converse}. Note that for most states the condition $|\rho_{ij}| = \sqrt{\rho_{ii} \rho_{jj}}$ is met only when $i=j$, which implies that $\widebar{\rho}=\Delta(\rho)$ and hence that $Q(\rho) =0$. This is a manifestation of the phenomenon of generic bound coherence discovered in~\cite{bound-coherence}, where it was also observed that the only states for which $Q(\rho)>0$ are those that admit a rank-deficient $2\times 2$ principal submatrix with strictly positive diagonal.

\subsection{PIO coherence cost} \label{subsec PIO cost}

The opposite process to coherence distillation is coherence dilution. Starting from a large supply of coherence bits, we want to prepare a large number of identical copies of a target state $\rho$ with vanishing error in the asymptotic limit. For a given class of operations $\pazocal{O}$, the optimal rate at which this can be accomplished is given by the \emph{coherence cost}, defined by
\bb
C_{c, \pazocal{O}}(\rho) \coloneqq \inf \left\{r: \lim_{n\to\infty} \inf_{\Lambda\in \pazocal{O}} \left\| \Lambda( \Psi_{2^{\floor{rn}}} )\! -\! \rho^{\otimes n} \right\|_1 = 0 \right\} .
\label{cost}
\ee

\begin{thm} \label{C c PIO thm}
For all states $\rho$, the coherence cost under PIO is given by the \emph{uniform coherence of formation}:
\bb
C_{c,\PIO} (\rho) = C_f^\M(\rho) \coloneqq \inf_{\substack{\\ \sum_\alpha p_\alpha \Psi_\alpha = \rho \\[.5ex] \ket{\Psi_\alpha}\in\M_{k_\alpha}}} \sum_\alpha p_\alpha \log k_\alpha\, ,
\label{PIO cost}
\ee
where the infimum runs over all decompositions of $\rho$ as a convex combinations of uniformly coherent states $\Psi_\alpha$ of size $k_\alpha$, and is set to be infinite if no such decomposition exists.
\end{thm}

For the proof we refer the reader to Section~\ref{sec PIO cost}. The above result quantifies exactly the power of the PIO class in the process of coherence dilution. In particular, it can be used to show that there is a ball around the maximally mixed state that is entirely formed by states with finite cost. More precisely, all states $\rho$ such that $\left\|\rho - \id/d\right\|_{1\to 1}\leq 1/d$, where $\|X\|_{1 \to 1}\coloneqq \max_{1\leq i\leq d} \sum_{j=1}^d |X_{ij}|$ is the max-row sum norm, satisfy $C_{c, \PIO}(\rho)\leq 1$. On the other hand, it is easy to verify that the uniform coherence of formation in infinite on all pure states that are not uniformly coherent. Consequently, these cannot be prepared via PIO starting from any number of coherence bits, which recovers one of the results in~\cite{Chitambar-PIO}.

\section{Distillable coherence under SIO/PIO: the protocol} \label{sec direct}

The first step in proving Theorem~\ref{solution thm} is to show the achievability of the quintessential coherence $Q$ as an SIO distillation rate. To this end, throughout this section we will recap the SIO distillation protocol constructed in~\cite{bound-coherence}, which allows us to distil coherence at rate $Q$ as required. Let $\rho$ be a quantum state in dimension $d$. We construct the positive semidefinite matrix
\bb
R^\rho \coloneqq \Delta(\rho)^{-1/2} \rho\, \Delta(\rho)^{-1/2} ,
\label{R rho}
\ee
where the inverse of $\Delta(\rho)$ (the diagonal part of $\rho$) is taken on its support. Observe that $R^\rho_{ii} = 1$ iff $\rho_{ii}>0$. Consider the graph $G_\rho = (V_\rho, E_\rho)$ with vertices $V_\rho \coloneqq [d]$ and edges
\bb
E_\rho \coloneqq \left\{ (i,j)\!: |R^\rho_{ij}|=1 \right\} = \left\{ (i,j)\!: |\rho_{ij}|=\! \sqrt{\rho_{ii}\rho_{jj}} >0 \right\} .
\label{E rho}
\ee
For simplicity, we have included into $E_\rho$ also diagonal pairs of the form $(i,i)$, with $i$ satisfying $\rho_{ii}>0$. The fact that $\rho$ is positive semidefinite has some strong implications for the structure of the above graph~\cite[SM,~Lemma~4]{bound-coherence}.

\begin{note}
From now on, we will often assume that $\Delta(\rho)>0$ has full support. This simplifies the notation considerably and causes no loss of generality, because the support of $\rho$ is necessarily contained inside $\Span\left\{\ket{i}:\rho_{ii}>0\right\}$.
\end{note}

\begin{lemma} \emph{\cite[SM,~Lemma~4]{bound-coherence}.} \label{G rho cliques lemma}
The connected components of the graph $G_\rho$ are all cliques (i.e.\ complete subgraphs). Equivalently, there exists a partition $\{I_s\}_{s\in\pazocal{S}}$ of $[d]$ such that 
\bb
(i,j)\in E_\rho \quad \Longleftrightarrow\quad \exists\ s\in\pazocal{S}:\ i,j\in I_s\, .
\label{E rho partition}
\ee
Moreover, for all $s\in\pazocal{S}$ the state $\frac{\Pi_{I_s}\rho\Pi_{I_s}}{\Tr[\Pi_{I_s}\rho]}$ is pure.
\end{lemma}

In~\cite{bound-coherence} it was shown that in order for a state to be SIO distillable there need to exists two indices $i\neq j$ such that $|\rho_{ij}| =\sqrt{\rho_{ii}\rho_{jj}}>0$. Intuitively, this seems to suggest that the only coherence inside $\rho$ that truly matters as far as SIO distillation is concerned is that identified by the entries $\rho_{ij}$ corresponding to pairs $(i,j)\in E_\rho$. We can thus construct a `trimmed' state $\widebar{\rho}$ by cutting off all other entries:
\bb
\widebar{\rho} \coloneqq \sum_{\substack{i,j=1,\ldots,d\\[.4ex] (i,j)\in E_\rho}} \rho_{ij} \ketbraa{i}{j} = \sum_{s\in \pazocal{S}} \Pi_{I_s} \rho\, \Pi_{I_s}\, ,
\label{rho bar}
\ee
where the sets $I_s$ are those identified by Lemma~\ref{G rho cliques lemma}, and for $I\subseteq [d]$ we define as usual $\Pi_I\coloneqq \sum_{i\in I} \ketbra{i}$. The second equality in Eq.~\eqref{rho bar} is a direct consequence of Lemma~\ref{G rho cliques lemma}, and implies -- among other things -- that $\widebar{\rho}$ is positive semidefinite and thus a legitimate density matrix (normalisation follows easily as $\Delta(\widebar{\rho}) = \Delta(\rho)$). This line of thought leads us to define the \emph{quintessential coherence} as
\bb
Q(\rho)\coloneqq S\left( \Delta(\rho)\right) - S(\widebar{\rho})\, .
\label{Q}
\ee

Since most states are such that all $2\times 2$ principal minors are strictly positive, and this property implies that $\widebar{\rho} = \Delta(\rho)$, the quintessential coherence vanishes on all but zero-measure sets of states, in compliance with the results of~\cite{bound-coherence}. In particular, $Q$ is highly discontinuous. It is shown in~\cite[SM,~Lemma~7]{bound-coherence} that $Q$ is fully additive over tensor products, i.e.
\bb
Q(\rho\otimes \sigma) = Q(\rho) + Q(\sigma)
\label{additivity Q}
\ee
for all states $\rho,\sigma$. 

\begin{rem} \label{Q and Cr and Cf rem}
It is worth noticing that coherence of formation and relative entropy of coherence coincide precisely for states such that $\widebar{\rho}=\rho$~\cite[Theorem~10]{Winter2016}. This implies that
\bb
Q(\rho) = C_r(\widebar{\rho}) = C_f(\widebar{\rho})
\label{Q and Cr and Cf}
\ee
holds for all $\rho$.
\end{rem}

Although it is not clear at first sight, the quintessential coherence is an SIO monotone, as will follow from Theorem~\ref{solution thm} once we have proved it in Section~\ref{sec converse}. We now recall the result presented in~\cite{bound-coherence} that $Q(\rho)$ is at least an achievable rate for SIO distillation, which establishes the direct part of Theorem~\ref{solution thm}. For the sake of readability, we include a sketch of the proof reported in~\cite[SM]{bound-coherence}.

\begin{lemma} \label{C d SIO lower bound lemma}
The SIO/PIO distillable coherences satisfy
\bb
C_{d,\SIO}(\rho) \geq C_{d,\PIO}(\rho) \geq Q(\rho)
\label{C d SIO lower bound}
\ee
for all states $\rho$.
\end{lemma}

\begin{proof}
On the one hand, since PIO is a subset of SIO, from Eq.~\eqref{distillable} it follows easily that $C_{d,\SIO}(\rho) \geq C_{d,\PIO}(\rho)$. On the other hand, there is a simple PIO protocol that produces an average of $nQ(\rho)$ coherence bits starting from $n$ identical copies of $\rho$. We describe and analyse it in intuitive terms here, referring to~\cite[SM]{bound-coherence} for a more rigorous analysis. There are three main steps.
\begin{enumerate}[(i)]
\item One applies the PIO instrument with Kraus operators $\{\Pi_{I_s}\}_{s\in\pazocal{S}}$ on each of the $n$ copies of $\rho$ that are initially available.
\item In the limit of large $n$, each outcome $s$ is obtained an average number of times equal to $nP(s)$, where $P(s) = \Tr [\Pi_{I_s}\rho]$. 
\item The post-measurement state corresponding to the outcome $s$, denoted by $\widetilde{\rho}_s\coloneqq P(s)^{-1} \Pi_{I_s}\rho \Pi_{I_s}$, is pure, as follows from Lemma~\ref{G rho cliques lemma}; it is then known (\cite{Winter2016} and~\cite[Proposition~7]{Chitambar2016}) that there is a PIO protocol that extracts coherence bits at a rate $S\left(\Delta(\widetilde{\rho}_s)\right)$; since we started with $nP(s)$ states, we obtain $n P(s) S\left(\Delta(\widetilde{\rho}_s)\right)$ cosbits at the output.
\end{enumerate}
The distillation rate associated with this protocol is then
\begin{align*}
r &= \sum_{s \in \pazocal{S}} P(s)\, S\left(\Delta(\widetilde{\rho}_s)\right) \\
&= \sum_{s \in \pazocal{S}} P(s)\, S\left(P(s)^{-1} \Pi_{I_s} \Delta(\rho) \Pi_{I_s} \right) \\
&= S(\Delta(\widebar{\rho})) - S(\widebar{\rho}) \\
&= S(\Delta(\rho)) - S(\widebar{\rho}) \\
&= Q(\rho)\, ,
\end{align*}
as claimed.
\end{proof}

\begin{rem}
We cannot improve the above distillation protocol by applying it to multiple copies of $\rho$. In fact, the identity $\frac1n Q(\rho^{\otimes n})=Q(\rho)$, which descends from Eq.~\eqref{additivity Q}, shows that the resulting rate would not be greater than that of the single-copy scenario.
\end{rem}

\section{A family of SIO monotones} \label{sec monotones}

Throughout this section, we will construct a family of SIO monotones and study their properties. These tools will eventually allow us to show in Section~\ref{sec converse} that the quintessential coherence $Q$ is also an upper bound to the SIO distillable coherence, completing the proof of Theorem~\ref{solution thm}.

\subsection{Definitions and elementary properties}

Let $\rho$ be a $d$-dimensional quantum state such that $\Delta(\rho)>0$. For an arbitrary integer $1\leq k\leq d$, define
\bb
\mu_k(\rho) \coloneqq \max_{I\subseteq [d],\, |I|\leq k} \log \left\|\Pi_I R^\rho \Pi_I \right\|_\infty\, ,
\label{mu}
\ee
where $R^\rho$ is given by Eq.~\eqref{R rho}, $\|\cdot\|_\infty$ denotes the operator norm, and again $\Pi_I = \sum_{i\in I} \ketbra{i}$. The maximum is achieved when $|I|=k$. 
Observe that for all states $\rho$ it holds that $\mu_1(\rho)\equiv 0$, while $\mu_2(\rho)=\log(1+\eta(\rho))$ is a function of the \emph{maximal coherence}~\cite{bound-coherence}, denoted by $\eta$ and defined in the forthcoming Eq.~\eqref{eta}. Using the fact that $[\Pi_I, D]\equiv 0$ for all diagonal $D$, one can also show that
\bb
\mu_k(\rho) = \max_{I\subseteq [d],\, |I|\leq k} D_{\max}\left( \Pi_I \rho \Pi_I \big\| \Delta(\rho)\right) ,
\label{mu D max}
\ee
where assuming that $\supp \sigma\subseteq \supp \omega$ the quantum max-relative entropy between $\sigma$ and $\omega$ is given by~\cite{Datta2009} 
\bb
\begin{aligned}
D_{\max}(\sigma\|\omega)\coloneqq &\, \inf\left\{\nu: \sigma\leq 2^{\nu} \omega \right\} \\
=&\, \log \left\| \omega^{-1/2}\sigma\,\omega^{-1/2}\right\|_\infty \, .
\label{D max}
\end{aligned}
\ee
Note that the inverse of $\omega$ is taken as usual on its support. The following lemma collects all elementary properties of the functions $\mu_k$.

\begin{lemma} \label{properties mu lemma}
Let $1\leq k\leq d$ be fixed. Then:
\begin{enumerate}[(a)]
\item $0\leq \mu_k (\rho)\leq \log k$ for all states $\rho$;
\item $\mu_k$ admits the following variational characterisation:
\bb
\mu_k (\rho) = \inf\{ \nu: \Pi_I \rho \Pi_I \leq 2^\nu \Delta(\rho) \ \forall\, I\! \subseteq \! [d]: |I|\!\leq\! k \}\, ;
\label{mu variational}
\ee
\item $\mu_k$ is an SIO monotone;
\item $\mu_k$ is lower semicontinuous.
\end{enumerate}
\end{lemma}

\begin{proof}
The upper estimate in~(a) can be deduced by remembering that for $k\times k$ positive matrices $A\geq 0$ the inequality $A\leq k \Delta(A)$ holds. Applying this to $A=\Pi_I R^\rho \Pi_I$ for some $I\subseteq [d]$ with $|I|=k$ and remembering that $R_{ii}^\rho\equiv 1$ for all $i$ yields $\Pi_I R^\rho \Pi_I \leq k \Pi_I$, implying that $\left\|\Pi_I R^\rho \Pi_I \right\|_\infty \leq k$.

Property~(b) can be deduced by putting together Eq.~\eqref{mu D max} and the variational representation in Eq.~\eqref{D max}, and immediately implies~(d). It is thus left to show~(c). Using the Kraus representation in Eq.~\eqref{SIO Kraus} for an SIO $\Lambda$ together with Eq.~\eqref{mu variational}, for any fixed $I\subseteq [d]$ such that $|I|=k$ we can write
\begin{align*}
\Pi_I \Lambda(\rho) \Pi_I &= \sum_\alpha \Pi_I U_{\pi_\alpha} D_\alpha \rho D_\alpha^* U_{\pi_\alpha}^\intercal \Pi_I \\
&= \sum_\alpha U_{\pi_\alpha} \Pi_{\pi_\alpha^{-1}(I)} D_\alpha \rho D_\alpha^* \Pi_{\pi_\alpha^{-1}(I)} U_{\pi_\alpha}^\intercal \\
&= \sum_\alpha U_{\pi_\alpha} D_\alpha \Pi_{\pi_\alpha^{-1}(I)} \rho\, \Pi_{\pi_\alpha^{-1}(I)} D_\alpha^* U_{\pi_\alpha}^\intercal \\
&\leq 2^{\mu_k(\rho)} \sum_\alpha U_{\pi_\alpha} D_\alpha \Delta(\rho) D_\alpha^* U_{\pi_\alpha}^\intercal \\
&= 2^{\mu_k(\rho)} \Delta\left( \sumno_\alpha U_{\pi_\alpha} D_\alpha \rho D_\alpha^* U_{\pi_\alpha}^\intercal\right) \\
&= 2^{\mu_k(\rho)} \Delta\left( \Lambda(\rho)\right) .
\end{align*}
Using once again Eq.~\eqref{mu variational}, we deduce that $\mu_k\left(\Lambda(\rho)\right)\leq \mu_k(\rho)$, which concludes the proof.
\end{proof}


\subsection{Some technical lemmata}

Before we proceed to explore some applications, we present two technical lemmata that will help us to evaluate the functions $\mu_k$ in certain circumstances. The first estimate deals with the case of a state that resembles closely a maximally coherent state.

\begin{lemma} \label{mu almost max coh lemma}
For all states $\rho$ in dimension $d$, one has that
\bb
\mu_d(\rho)\geq \log d + \log \braket{\Psi_d|\rho|\Psi_d}\, ,
\label{mu almost max coh}
\ee
for all maximally coherent states $\ket{\Psi_d}$ of size $d$.
\end{lemma}

\begin{proof}
To estimate the norm $\left\|R^\rho\right\|_\infty$ we evaluate the overlap of the operator $R^\rho$ with the normalised vector $\sqrt{d\, \Delta(\rho)}\ket{\Psi_d}$. A simple computation yields
\begin{align*}
\mu_d(\rho) &= \log \left\| R^\rho \right\|_\infty \\
&= \log \left\| \Delta(\rho)^{-1/2}\rho\, \Delta(\rho)^{-1/2} \right\|_\infty \\
&\geq \log d \braket{\Psi_d | \Delta(\rho)^{1/2} \!\Delta(\rho)^{-1/2}\rho\, \Delta(\rho)^{-1/2} \!\Delta(\rho)^{1/2} | \Psi_d} \\
&= \log  d \braket{\Psi_d | \rho | \Psi_d} \\
&= \log d + \log \braket{\Psi_d | \rho | \Psi_d}\, ,
\end{align*}
as claimed.
\end{proof}

We now want to establish an upper bound to quantify the intuitive fact that $\mu_k(\rho)$ grows slower than $\log k$ when $k$ becomes larger than the maximal size of a rank-one principal submatrix of $\rho$ with non-vanishing diagonal, denoted by $l(\rho)$:
\bb
l(\rho) \coloneqq \max\left\{ \rk \left(\Pi_I \Delta(\rho) \Pi_I \right)\!:\, I\!\subseteq\! [d],\, \rk\left( \Pi_I \rho\Pi_I \right) = 1 \right\} .
\label{l}
\ee
The following alternate characterisation of $l(\rho)$ will be helpful in the following.

\begin{lemma} \label{l max 1s row lemma}
Let $\rho$ be a state in dimension $d$ such that $\Delta(\rho)>0$. The function $l(\rho)$ in~\eqref{l} can also be computed as the maximum number of entries of modulus $1$ in a row of the matrix $R^\rho$ in~\eqref{R rho}. In formula,
\bb
l(\rho) = \max_{i\in [d]} \left| \left\{ j\in [d]:\, | R^\rho_{ij}| = 1 \right\}\right| .
\label{l max 1s row}
\ee
\end{lemma}

\begin{proof}
Call $l'(\rho)$ the quantity defined by the r.h.s.\ of~\eqref{l max 1s row}. Use Lemma~\ref{G rho cliques lemma} to construct a partition $\{I_s\}_{s\in\pazocal{S}}$ of $[d]$ such that $|R^\rho_{ij}| =1$ if and only if there exists $s\in\pazocal{S}$ such that $i,j\in I_s$. It follows immediately that
\bbb
l'(\rho) = \max_{s\in\pazocal{S}} |I_s|\, .
\eee
Since $\Delta(\rho)>0$ is invertible and commutes with any incoherent projector $\Pi_I$, we have that $\rk\left( \Pi_I \rho\Pi_I\right) = \rk\left( \Pi_I R^\rho \Pi_I\right)$ and $\rk\left( \Pi_I \Delta(\rho)\Pi_I\right) = |I|$. The second claim of Lemma~\ref{G rho cliques lemma} guarantees that $\rk\left( \Pi_{I_s} \rho\Pi_{I_s}\right) = \rk\left( \Pi_{I_s} R^\rho\Pi_{I_s}\right) =1$ for all $s\in\pazocal{S}$. Comparing this with~\eqref{l} we deduce that $l(\rho)\geq \max_{s\in\pazocal{S}} |I_s| = l'(\rho)$. To finish the proof we need only to show that $l(\rho)\leq l'(\rho)$. This descends from the fact that if $I\subseteq [d]$ satisfies $\rk\left( \Pi_I\rho\Pi_I\right) = \rk\left( \Pi_I R^\rho \Pi_I \right) =1$ then necessarily $|R^\rho_{ij}|=1$ for all $i,j\in I$, hence $|I|\leq l'(\rho)$.
\end{proof}

Another useful definition is as follows:
\bb
\lambda(\rho) \coloneqq \max\left\{ |R^\rho_{ij}|:\ 1\leq i<j\leq d,\ |R^\rho_{ij}|<1 \right\} ,
\label{lambda}
\ee
where we put $\lambda(\rho)=0$ if the set on the r.h.s.\ is empty (which happens iff $\rho$ is pure and $\Delta(\rho)>0$). The quantity $\lambda$ is closely related to the maximal coherence $\eta$ introduced in~\cite{bound-coherence}:
\bb
\eta(\rho) \coloneqq \max\left\{ |R^\rho_{ij}|:\ 1\leq i<j\leq d \right\} .
\label{eta}
\ee
By looking at the two definitions it is easy to see that: (i)~$0\leq \lambda(\rho)\leq \eta(\rho)\leq 1$ for all states $\rho$; (ii)~$\lambda(\rho)<1$; (iii)~it holds that $\eta(\rho) = \lambda(\rho)$ provided that $\eta(\rho)<1$, while there are examples of states for which $1=\eta(\rho)>\lambda(\rho)$. Moreover, (iv)~it holds that
\bb
\lambda(\rho)=0\quad \Longleftrightarrow\quad \rho =\widebar{\rho}\, ,
\label{lambda = 0}
\ee
while $\eta(\rho)=0$ iff $\rho=\Delta(\rho)$. Finally, (v)~maximal and quintessential coherence are related by the fact that $\eta(\rho)<1$ iff $Q(\rho)=0$. It is maybe less straightforward to see that $\lambda$ exhibits a `tensorisation property' very similar to that satisfied by $\eta$ and proven in~\cite{bound-coherence}.

\begin{lemma} \label{lambda tensorisation lemma}
For all pairs of states $\rho,\sigma$, the quantifier $\lambda$ of Eq.~\eqref{lambda} obeys the following tensorisation property:
\bb
\lambda(\rho\otimes\sigma) = \max\{ \lambda(\rho), \lambda(\sigma) \} .
\label{lambda tensorisation}
\ee
\end{lemma}

\begin{proof}
The argument is very similar to that presented in~\cite{bound-coherence} for the maximal coherence. Assume without loss of generality that $\Delta(\rho)$ and $\Delta(\sigma)$, albeit of possibly different sizes $d$ and $d'$, are both invertible. Then rows and columns of $\rho\otimes \sigma$ are indexed by pairs $(i,l)$, where $1\leq i\leq d$ and $1\leq l\leq d'$, and $(i,l)\neq (j,m)$ iff either $i\neq j$ or $l\neq m$. Since $R^\rho_{ii}=R^\sigma_{ll}=1$ for all $i$ and $l$, the maximum of $\big|R^{\rho\otimes \sigma}_{(i,l), (j,m)}\big| = |R^\rho_{ij}|\, |R^\sigma_{lm}|$ over pairs $(i,l)\neq (j,m)$ is clearly achieved either when $i=j$ (yielding $\lambda(\sigma)$) or when $l=m$ (yielding $\lambda(\rho)$).

When the sets on the r.h.s.\ of Eq.~\eqref{lambda} are empty for both $\rho$ and $\sigma$, which are then pure, according to our conventions we have $\lambda(\rho)=\lambda(\sigma)=0=\lambda(\rho\otimes \sigma)$, where the last equality follows because also $\rho\otimes \sigma$ is pure.
\end{proof}

We can now prove the following.

\begin{lemma} \label{mu and l lemma}
For a $d$-dimensional state $\rho$ and all integers $1\leq k\leq d$ one has that
\bb
\mu_k(\rho)\leq \log \left[l(\rho) + \lambda(\rho) (k - l(\rho) ) \right] .
\label{mu and l}
\ee
\end{lemma}

\begin{proof}
When $k< l(\rho)$ the claim is trivial, because the r.h.s.\ of Eq.~\eqref{mu and l} is larger than $\log k$, and $\mu_k(\rho)\leq \log k$ always holds by Lemma~\ref{properties mu lemma}(a). In what follows we therefore assume that $k\geq l(\rho)$.

As usual, we can also suppose without loss of generality that $\Delta(\rho)>0$. Ger\v{s}gorin's theorem (\cite{Gershgorin} or~\cite[Theorem~6.1.1]{HJ1}) implies that all eigenvalues of $\Pi_I R^\rho \Pi_I$ lie in the region of the complex plane enclosed in a circle centred on $1$ and having radius 
\bbb
\max_i \sum_{j\neq i} \left|(\Pi_I R^\rho \Pi_I)_{ij}\right| = \max_{i\in I} \sum_{j\in I,\, j\neq i} |R^\rho_{ij}|\, .
\eee
Since $\mu_k(\rho)$ is nothing but the logarithm of the maximal eigenvalue of some $\Pi_I R^\rho \Pi_I$, we can estimate it as
\begin{align*}
\mu_k(\rho) &\leq \max_{|I|\leq k} \max_{i\in I} \log\left[1+\sumno_{j\in I,\, j\neq i} |R^\rho_{ij}|\right] \\
&\leq \max_{|I|\leq k} \max_{i\in I} \log\left[\sumno_{j\in I} |R^\rho_{ij}|\right] .
\end{align*}
If $k\geq l(\rho)$, by Lemma~\ref{l max 1s row lemma} in any fixed row of $R^\rho$ there are at most $l(\rho)$ entries of modulus $1$, while all others have modulus at most $\lambda(\rho)$. Hence, when $|I|\leq k$ and $i\in I$ one has that $\sumno_{j\in I} |R^\rho_{ij}|\leq l(\rho) + \lambda(\rho) (k-l(\rho))$, which inserted into the above estimate yields Eq.~\eqref{mu and l} and completes the proof.
\end{proof}

\subsection{Smoothing}

We now discuss smoothed versions of the monotones $\mu_k$ introduced in Eq.~\eqref{mu}. For a generic $\epsilon>0$, let us define
\bb
\mu_k^\epsilon(\rho)\coloneqq \min_{\sigma\in B_\epsilon(\rho)} \mu_k(\sigma)\, ,
\label{mu epsilon}
\ee
where $B_\epsilon(\rho)$ is the set of states at trace norm distance at most $\epsilon$ from $\rho$, i.e.
\bb
B_\epsilon(\rho)\coloneqq \left\{ \sigma:\ \left\|\sigma - \rho\right\|_1\leq \epsilon \right\} ,
\label{ball B}
\ee
where $\sigma$ is a normalised density matrix.
Not surprisingly, the monotonicity of $\mu_k$ as established by Lemma~\ref{properties mu lemma}(c) ensures the following.

\begin{lemma} \label{mu epsilon monotone lemma}
For all positive integers $k$ and all $\epsilon>0$, the function $\mu_k^\epsilon$ in Eq.~\eqref{mu epsilon} is an SIO monotone.
\end{lemma}

\begin{proof}
Let $\rho$ be a state and $\Lambda$ an SIO. Since quantum channels never increase the trace norm, one has that $\Lambda\left( B_\epsilon (\rho) \right) \subseteq B_\epsilon \left( \Lambda(\rho) \right)$. Using the monotonicity of $\mu_k$ under $\Lambda$ (Lemma~\ref{properties mu lemma}(c)), one obtains that
\begin{align*}
\mu_k^\epsilon \left( \Lambda(\rho)\right) &= \min_{\sigma\in B_\epsilon \left( \Lambda(\rho) \right)} \mu_k ( \sigma) \\
&\leq \min_{\sigma\in \Lambda\left( B_\epsilon (\rho) \right)} \mu_k ( \sigma) \\
&= \min_{\omega \in B_\epsilon (\rho) } \mu_k \left( \Lambda(\omega) \right) \\
&\leq \min_{\omega \in B_\epsilon (\rho) } \mu_k \left( \omega \right) \\
&= \mu_k^\epsilon(\rho)\, ,
\end{align*}
proving the claim.
\end{proof}

\section{Distillable coherence under SIO/PIO: converse} \label{sec converse}

The purpose of this section is to prove the converse part of Theorem~\ref{solution thm}, determining the SIO distillable coherence for all states and showing that it coincides with the quintessential coherence of Eq.~\eqref{Q}. The monotones $\mu_k$ and more precisely they smoothed versions $\mu_k^\epsilon$ will play a crucial role in our argument.

\subsection{Preliminaries}

We start by setting some notation. Given a state $\rho$, consider the family $\left\{I_s^\rho\right\}_{s\in\pazocal{S}^\rho}$ of disjoint subsets of $[d]$ that is associated with it via Lemma~\ref{G rho cliques lemma} (this is a partition of $[d]$ provided that $\Delta(\rho)>0$). Observe that we added a superscript to indicate its dependence on $\rho$. For any other state $\sigma$, we can then construct a random variable $S^\rho_\sigma$ whose probability distribution takes the form $P_{S^\rho_\sigma}(s)\coloneqq \Tr\left[\sigma \Pi_{I_s^\rho}\right]$. Clearly, $S^\rho_\sigma$ is a coarse-grained version of the random variable $J_\sigma$ with range $[d]$ distributed according to $P_{J_\sigma} (j)\coloneqq \sigma_{jj} = \braket{j|\sigma|j}$. A first important observation is that the quintessential coherence defined in Eq.~\eqref{Q} coincides with the conditional entropy of $J_\rho$ given $S^\rho_\rho$, in formula 
\bb
Q(\rho) = H(J_\rho |S^\rho_\rho) = H(J|S^\rho)_{\delta_\rho}\, .
\label{Q entropies}
\ee
Here, the rightmost side refers to the conditional entropy of $J$ given $S^\rho$ as computed on the probability distribution $\delta_\rho$ on the set $[d]$ defined by $\delta_\rho(j)\coloneqq \rho_{jj}$.
To see that~\eqref{Q entropies} holds, start by observing that by construction $S(\Delta(\rho)) = H(J_\rho)=H(J_\rho S^\rho_\rho)$, where the last equality follows because $S^\rho_\rho$ is a deterministic function of $J_\rho$. Moreover, because of Lemma~\ref{G rho cliques lemma} the state $\widebar{\rho}$ in~\eqref{rho bar} is block-diagonal, and each block is proportional to a pure state. Hence, its entropy evaluates to
\begin{align*}
S(\widebar{\rho}) &= S\left( \sumno_{s\in\pazocal{S}} \Pi_{I_s}\rho\Pi_{I_s}\right) \\
&= - \sum_{s\in\pazocal{S}} \Tr\left[ \Pi_{I_s} \rho\Pi_{I_s} \log\left( \Pi_{I_s}\rho\Pi_{I_s}\right)\right] \\
&= - \sum_{s\in\pazocal{S}} \Tr\left[ \rho \Pi_{I_s} \right] \log\Tr\left[\rho \Pi_{I_s} \right] \\
&= H(S^\rho_\rho) \, .
\end{align*}

Also the function $l$ in Eq.~\eqref{l} can be expressed in terms of these entropies. Namely, it is not difficult to show that
\bb
\log l(\rho) = H_{\max}(J_\rho |S^\rho_\rho) = H_{\max}(J|S^\rho)_{\delta_\rho}\, ,
\label{l entropies}
\ee
where for two classical random variables $X,Y$ with probability distribution $p=p_{XY}$ their conditional max entropy is given by $H_{\max}(X|Y) \coloneqq \max_{y} \log  | \supp p_{X|y} |$, with $p_{X|y}$ being the probability distribution of $X$ conditioned on $Y=y$, and $\supp$ denoting the support.

In what follows, we will find it useful, to look at the set $V_\epsilon(\rho)$ defined for a generic $\epsilon>0$ as
\bb
V_\epsilon(\rho) \coloneqq \left\{ \frac{\Pi_I \rho\Pi_I}{\Tr\left[\rho \Pi_I\right]}:\ I\subseteq [d] \right\} \cap B_\epsilon(\rho)
\label{V epsilon quantum}
\ee
in terms of the trace norm balls in Eq.~\eqref{ball B}. Via the gentle measurement lemma~\cite[Lemma~9]{VV1999}, $V_\epsilon(\rho)$ can be shown to include all post-measurement states obtained by making a binary incoherent measurement whose success probability on $\rho$ is sufficiently close to $1$, with the condition that said measurement has been successful. Observe that $V_\epsilon(\rho)$ is not a ball in the proper sense; indeed, it is always a finite set, and moreover $V_\epsilon(\rho)=\{\rho\}$ for all sufficiently small $\epsilon$ provided that $\Delta(\rho)>0$.

However, the following important features of $V_\epsilon(\rho)$ make it very important for applications.

\begin{enumerate}[(i)]

\item For all states $\sigma\in V_\epsilon(\rho)$, the families $\{I_s^\sigma\}_{s\in\pazocal{S}^\sigma}$ associated to them via Lemma~\ref{G rho cliques lemma} are very similar to each other. Namely,
\bb
\pazocal{S}^\sigma\subseteq \pazocal{S}^\rho \qquad \text{and} \qquad I^\sigma_s\subseteq I^\rho_s\quad \forall\ s\in\pazocal{S}^\sigma\, .
\ee
This practically implies that
\bb
\sigma\in V_\epsilon(\rho) \quad \Longrightarrow\quad S^\rho_\sigma = S^\sigma_\sigma\, ,
\label{crucial property V epsilon}
\ee
in the sense that the two random variables have the same effective range and the same probability distribution.

\item The monotone $\lambda$ defined in Eq.~\eqref{lambda} is also very well-behaved on the sets $V_\epsilon(\rho)$. Namely, it is not difficult to verify that
\bb
\sigma\in V_\epsilon(\rho) \quad \Longrightarrow\quad \lambda(\sigma)\leq \lambda(\rho)\, .
\label{lambda on V epsilon}
\ee
In fact, since $\sigma\propto \Pi_I\rho\Pi_I$ for some $I\subseteq [d]$:
\begin{align*}
\lambda(\sigma) &= \max\left\{ | R^\sigma_{ij} | :\ 1\leq i<j\leq d,\ | R^\sigma_{ij} |<1 \right\} \\
&= \max\left\{ | R^\rho_{ij} | :\ i,j\in I,\ i\neq j,\ | R^\rho_{ij} | < 1 \right\} \\
&\leq \max\left\{ | R^\rho_{ij} | :\ 1\leq i<j\leq d,\ | R^\rho_{ij} |<1 \right\} \\
&= \lambda(\rho)\, .
\end{align*}
Observe that the above inequality remains valid also when there are no pairs $(i,j)$ satisfying $|R^\sigma_{ij}|<1$. Indeed, in that case $\sigma$ is necessarily pure, and we set by convention $\lambda(\sigma)=0$, while $\lambda(\rho)\geq 0$ always holds by construction. 

\end{enumerate}

Although until now we have been concerned mostly with the quantum case, the sets $V_\epsilon$ can be defined in pretty much the same way in the classical setting as well. Namely, for an arbitrary probability distribution $p$ on the set $[d]$, intended as a vector $p\in \R^d$, one can construct
\bb
V_\epsilon(p) \coloneqq \left\{ \frac{\Pi_I \, p}{\sum_{i\in I} p_i}:\ I\subseteq [d] \right\} \cap B_\epsilon(p)\, ,
\label{V epsilon classical}
\ee
where $B_\epsilon(p)\coloneqq \left\{ q\in \R^d: |p-q|_1\leq \epsilon\right\}$, with $|\cdot|_1$ being the $\ell_1$-norm. 

The classical and quantum constructions are closely related to each other. To make this statement precise, consider an optimisation problem of the form $\min_{\sigma \in V_\epsilon(\rho)} f(\delta_\sigma)$, where $f$ is a real-valued function defined on the set of probability distributions over $d$ elements. We could try to compare this to its fully classical version $\min_{q\in V_\epsilon(\delta_\rho)} f(q)$. The following lemma shows that this is in some sense possible, indeed.

\begin{lemma} \label{optimisation comparison lemma}
For all states $\rho$, all $\epsilon>0$, and all real-valued functions $f$ defined on the set of probability distributions over $d$ elements, it holds that
\bbb
\min_{q \in V_{\epsilon}(\delta_\rho)} f(q) \leq \min_{\sigma \in V_\epsilon(\rho)} f(\delta_\sigma) \leq \min_{q \in V_{\epsilon^2/4}(\delta_\rho)} f(q) \, ,
\eee
where as usual $\delta_\omega$ denotes the diagonal of a $d$-dimensional quantum state $\omega$, intended as a vector in $\R^d$.
\end{lemma}

\begin{proof}
The first inequality follows trivially from the fact that 
\bbb
\left| \delta_\sigma - \delta_\rho \right|_1 = \left\|\Delta(\sigma) - \Delta(\rho)\right\|_1 \leq \left\|\rho - \sigma\right\|_1
\eee
because $\Delta(\cdot)$ is a quantum channel. Then, for all $\sigma=\frac{\Pi_I \rho \Pi_I}{\Tr[\rho \Pi_I]}\in V_\epsilon(\rho)$ one has that $\delta_\sigma = \frac{\Pi_I\, \delta_\rho}{\sum_{i\in I} \rho_{ii}}\in V_\epsilon(\delta_\rho)$, implying that $\min_{\sigma \in V_\epsilon(\rho)} f(\delta_\sigma)\geq \min_{q \in V_{\epsilon}(\delta_\rho)} f(q)$.

The second inequality is slightly less straightforward. Given $q = \frac{\Pi_I\, \delta_\rho}{\sum_{i\in I} \rho_{ii}} \in V_{\epsilon^2/4}(\delta_\rho)$, set
\bbb
\sigma \coloneqq \frac{\Pi_I \rho \Pi_I}{\Tr[\rho \Pi_I]}\, ,
\eee
so that $\delta_\sigma=q$. One has that
\begin{align*}
\Tr[\rho \Pi_I] &= \sum_{i\in I} \rho_{ii} \\
&= 1 - \sum_{i\in I} \left| q_{i} - (\delta_\rho)_i\right| \\
&\geq 1 - \left| q - \delta_\rho \right|_1 \\
&\geq 1-\frac{\epsilon^2}{4}\, .
\end{align*}

The gentle measurement lemma (see~\cite[Lemma~9]{VV1999} for the original version, and~\cite[Lemma~9.4.1]{MARK} for the one we use here) then ensures that
\bbb
\left\| \rho - \sigma\right\|_1 = \left\| \rho - \frac{\Pi_I \rho \Pi_I}{\Tr[\rho \Pi_I]}\right\|_1 \leq 2\sqrt{\frac{\epsilon^2}{4}} = \epsilon\, ,
\eee
i.e.\ $\sigma \in V_{\epsilon}(\rho)$. Hence, $\min_{q \in V_{\epsilon^2/4}(\delta_\rho)} f(q) \geq \min_{\sigma \in V_\epsilon(\rho)} f(\delta_\sigma)$.
\end{proof}

\subsection{A tweaked asymptotic equipartition property}

The standard \emph{smoothed conditional max entropy} can be defined for a pair of random variables $XY$ distributed according to $p$ as
\bb
H_{\max}^\epsilon \left(X|Y\right)_p \coloneqq \min_{q\in B_\epsilon(p)} H_{\max}\left(X|Y\right)_q\, .
\label{H max smoothed}
\ee
In terms of this quantity, the familiar form of the classical \emph{asymptotic equipartition property} (AEP) is the identity
\bb
\lim_{n\to \infty} \frac1n H_{\max}^\epsilon\left(X^n|Y^n\right)_{p^n} = H(X|Y)_p\qquad \forall\ \epsilon>0\, ,
\label{AEP standard}
\ee
where $X^n Y^n$ refers to $n$ i.i.d.\ copies of the pair of classical random variables $XY$ distributed according to $p$, the resulting product distribution being denoted with $p^n$. For a proof see for instance~\cite[Theorem~3.3.4 and Lemma~3.1.5]{RennerPhD}. Here we will not make use of Eq.~\eqref{AEP standard}. Instead, we will need a modified version of it, that features a minimisation not over $B_\epsilon(p)$ but over the smaller set $V_\epsilon(p)$ of Eq.~\eqref{V epsilon classical}. We thus define
\bb
\widetilde{H}_{\max}^\epsilon \left(X|Y\right)_p \coloneqq \min_{q\in V_\epsilon(p)} H_{\max}\left(X|Y\right)_q\, .
\label{H max tilde smoothed}
\ee

\begin{lemma}[Tweaked AEP] \label{tweaked AEP lemma}
For all pairs of classical random variables $XY$ distributed according to $p$, one has that
\bb
\lim_{n\to \infty} \frac1n \widetilde{H}_{\max}^\epsilon \left(X^n|Y^n\right)_{p^n} = H(X|Y)_p\qquad \forall\ \epsilon>0\, .
\label{tweaked AEP}
\ee
\end{lemma}

\begin{proof}
The statement could be derived from the results of~\cite{RennerPhD}, but the argument would be quite cumbersome while still requiring a significant amount of work. A direct proof is perhaps more transparent. We have to worry only about proving the upper bound in Eq.~\eqref{tweaked AEP}, as the inclusion $V_\epsilon(p^n)\subseteq B_\epsilon(p^n)$ automatically guarantees that
\begin{align*}
\lim_{n\to \infty} \frac1n \widetilde{H}_{\max}^\epsilon \left(X^n|Y^n\right)_{p^n} &\geq \lim_{n\to \infty} \frac1n H_{\max}^\epsilon \left(X^n|Y^n\right)_{p^n} \\
&= H(X|Y)_p
\end{align*}
for all $\epsilon>0$, where the last step is naturally an application of the standard equipartition property, Eq.~\eqref{AEP standard}.

In order to establish the converse bound, we start by introducing some notation. Consider a parameter $\delta>0$. For all sequences $y^n$, construct the \emph{weakly typical set}
\bbb
T_\delta^{X^n|y^n}\!\coloneqq \left\{x^n\! : \left| - \frac1n \log p_{X^n|Y^n}(x^n|y^n) - H(X|Y)\right|\leq \delta \right\}\! .
\eee
A survey of the main properties of this object can be found for instance in~\cite[\S~14.6.1]{MARK}. We will make use of the following two facts:
\begin{align}
&\ \log \left| T_\delta^{X^n|y^n}\right| \leq n \left( H(X|Y) + \delta\right) , \label{size typical set} \\[1ex]
&\lim_{n\to \infty} \pr_{X^nY^n} \left\{ x^n \in T_\delta^{X^n|y^n} \right\} = 1\qquad \forall\ \delta>0\, . \label{probability typical set}
\end{align}

Now, let $\epsilon>0$ be fixed. We have to show that for all $\delta>0$ there exists $N\in\N$ such that $\frac1n \widetilde{H}_{\max}^\epsilon \left(X^n|Y^n\right)_{p^n}\leq H(X|Y)+ \delta$ for all $n\geq N$. For all $n$, set
\bbb
I_\delta^{X^nY^n}\coloneqq \left\{ x^ny^n:\ x^n\in T_\delta^{X^n|y^n} \right\}\, .
\eee
Observe that Eq.~\eqref{probability typical set} can be rephrased by saying that $\lim_{n\to \infty} \pr \left( I_\delta^{X^nY^n} \right)= 1$, where it is understood that the probabilities are computed according to the distribution $p^n$. Let us pick $N\in\N$ such that
\bbb
\pr \left( I_\delta^{X^nY^n} \right) \geq 1-\frac{\epsilon}{2}\qquad \forall\ n\geq N\, .
\eee
Define
\bbb
q\coloneqq \frac{\Pi_{I_\delta^{X^nY^n}}\ p^n}{\pr \left( I_\delta^{X^nY^n} \right)}\, ,
\eee
so that
\begin{align*}
H_{\max}(X^n|Y^n)_q &= \max_{y^n} \log \left| \supp q_{X^n|y^n}\right| \\
&= \max_{y^n} \log \left| T_\delta^{X^n|y^n}\right| \\
&\leq n \left( H(X|Y) + \delta\right) ,
\end{align*}
where the last inequality comes from Eq.~\eqref{size typical set}.
Observe also that
\bbb
|q-p^n|_1 = 2 \left( 1 - \pr \left( I_\delta^{X^nY^n} \right) \right) \leq \epsilon\, ,
\eee
implying that
\bbb
q\in V_\epsilon(p^n)\, .
\eee
Thus,
\begin{align*}
\frac1n \widetilde{H}_{\max}^\epsilon \left(X^n|Y^n\right)_{p^n} &\leq \frac1n H_{\max} \left(X^n|Y^n\right)_{q} \\
&\leq H(X|Y)+\delta\, .
\end{align*}
Since the above estimate holds for all $n\geq N$, this concludes the proof.
\end{proof}

\subsection{First constraints on achievable rates}

In Section~\ref{sec monotones} we have introduced and studied a wealth of SIO monotones, namely the functions $\mu_k$ of Eq.~\eqref{mu} and their smoothed versions $\mu_k^\epsilon$ defined in Eq.~\eqref{mu epsilon}. However, until now we have not used them to derive constraints on the achievable SIO distillation rates. The following result deals precisely with this problem.

\begin{prop} \label{upper bound r prop}
Let $r$ be an achievable rate for SIO coherence distillation starting from a state $\rho$ (in the sense of Eq.~\eqref{distillable}). Then for all $\epsilon>0$ it holds that
\bb
\liminf_{n\to \infty} \left\{ \mu_{2^{\floor{rn}}}^{\epsilon} \left( \rho^{\otimes n}\right) - \floor{rn} \right\} \geq \log (1-\epsilon/2)\, ;
\label{upper bound r with epsilon}
\ee
Thus,
\bb
\lim_{\epsilon\to 0^+} \liminf_{n\to \infty} \left\{ \mu_{2^{\floor{rn}}}^{\epsilon} \left( \rho^{\otimes n}\right) - \floor{rn} \right\} \geq 0\, .
\label{upper bound r}
\ee
\end{prop}

\begin{proof}
For a fixed $\epsilon>0$, if $r$ is an achievable rate there must exist a sequence of SIO transformations $\Lambda_n$ such that $\left\|\Lambda_n \left( \rho^{\otimes n}\right) - \Psi_{2^{\floor{rn}}}\right\|_1\leq \epsilon$ eventually in $n$. By the Fuchs--van de Graaf inequality~\cite{Fuchs1999}, slightly optimised to take advantage of the fact that one of the two states is pure~\cite[Exercise~9.21]{NC}, this implies that
\bbb
\braket{\Psi_{2^{\floor{rn}}}| \Lambda_n \left( \rho^{\otimes n}\right) | \Psi_{2^{\floor{rn}}}} \geq 1-\epsilon/2\, .
\eee
Thanks to Lemma~\ref{mu almost max coh lemma}, we can then write
\begin{align*}
\mu_{2^{\floor{rn}}}^\epsilon\left(\rho^{\otimes n}\right) - \floor{rn} &\geq \mu_{2^{\floor{rn}}}^\epsilon\left(\Lambda_n\left(\rho^{\otimes n}\right)\right) - \floor{rn} \\[.8ex]
&\geq \log \braket{\Psi_{2^{\floor{rn}}}| \Lambda_n \left( \rho^{\otimes n}\right) | \Psi_{2^{\floor{rn}}}} \\[.8ex]
&\geq \log (1-\epsilon/2)\, .
\end{align*}
Since this holds eventually in $n$, we can take the $\liminf$ for $n\to\infty$ and obtain Eq.~\eqref{upper bound r with epsilon}. 
Computing the limit for $\epsilon\to 0^+$ yields Eq.~\eqref{upper bound r} and concludes the proof.
\end{proof}

\subsection{The converse bound}

We now shift the focus on the problem of finding tight upper bounds for $\mu_k^\epsilon \left(\rho^{\otimes n}\right)$. As it appears from an inspection of Proposition~\ref{upper bound r prop} and especially of Eq.~\eqref{upper bound r}, this will in turn give us upper bounds on the maximal achievable SIO distillation rate $r$.
Our approach to the problem will leverage the previously established Lemma~\ref{mu and l lemma}, whose proof rested on the beautiful theorem by Ger\v{s}gorin~\cite{Gershgorin}.

\begin{prop} \label{casino prop}
For all states $\rho$ such that $\lambda(\rho)>0$ and all pairs of positive integers $n,k$, it holds that
\bb
\begin{aligned}
\mu^\epsilon_k \left(\rho^{\otimes n}\right) &\leq \log k + \log \lambda(\rho) \\
&\quad + \log \left[ 1 + \frac{2^{\widetilde{H}_{\max}^{\epsilon^2/4} \left( J^n |(S^{\rho})^n \right)_{\delta_\rho^n}}}{k \lambda(\rho) }\right] .
\end{aligned}
\ee
\end{prop}

\begin{proof}
We write:
\begin{align*}
\mu^\epsilon_k \left(\rho^{\otimes n}\right) &= \min_{\sigma \in B_\epsilon\left(\rho^{\otimes n}\right)} \mu_k (\sigma) \\
&\textleq{1} \min_{\sigma \in V_\epsilon\left(\rho^{\otimes n}\right)} \mu_k (\sigma) \\
&\textleq{2} \min_{\sigma \in V_\epsilon\left(\rho^{\otimes n}\right)} \log \left[ l(\sigma) + \lambda(\sigma) \left( k - l(\sigma)\right) \right] \\
&\leq \min_{\sigma \in V_\epsilon\left(\rho^{\otimes n}\right)} \log \left[ l(\sigma) + k \lambda(\sigma) \right] \\
&\textleq{3} \min_{\sigma \in V_\epsilon\left(\rho^{\otimes n}\right)} \log \left[ l(\sigma) + k \lambda\left(\rho^{\otimes n}\right) \right] \\
&\texteq{4} \min_{\sigma \in V_\epsilon\left(\rho^{\otimes n}\right)} \log \left[ l(\sigma) + k \lambda(\rho) \right] \\
&= \log k + \log \lambda(\rho) \\
&\quad + \log \left[1 + \frac{1}{k\lambda(\rho)}\, 2^{\min_{\sigma \in V_\epsilon\left(\rho^{\otimes n}\right)} \log l(\sigma)} \right] .
\end{align*}
The justification of the above reasoning is as follows. 1: Restricting the optimisation set does not decrease the minimum; 2: comes from Lemma~\ref{mu and l lemma}; 3: is an application of Eq.~\eqref{lambda on V epsilon}; finally, 4: follows from Lemma~\ref{lambda tensorisation lemma}.

We now look at the minimum appearing in the expression we just found.
\begin{align*}
\min_{\sigma \in V_\epsilon\left(\rho^{\otimes n}\right)} \log l(\sigma) &\texteq{5} \min_{\sigma \in V_\epsilon\left(\rho^{\otimes n}\right)} H_{\max} \left( J^n |(S^{\sigma})^n \right)_{\delta_\sigma} \\
&\texteq{6} \min_{\sigma \in V_\epsilon\left(\rho^{\otimes n}\right)} H_{\max}\left( J^n |(S^{\rho})^n \right)_{\delta_\sigma} \\
&\textleq{7} \min_{q \in V_{\epsilon^2/4}\left(\delta_\rho^n\right)} H_{\max}\left( J^n |(S^{\rho})^n \right)_{q} \\
&= \widetilde{H}^{\epsilon^2/4}_{\max}\left( J^n |(S^{\rho})^n \right)_{\delta_\rho^n}\, .
\end{align*}
These steps are explained as follows. 5: we used Eq.~\eqref{l entropies}; 6: we employed Eq.~\eqref{crucial property V epsilon}; 7: we exploited the second inequality in Lemma~\ref{optimisation comparison lemma}. Putting all together concludes the proof.
\end{proof}

We are finally ready to prove the first of our main results.

\setthmtag{\ref{solution thm}}
\begin{thm}
For all states $\rho$, the distillable coherence under SIO/PIO is given by the quintessential coherence of Eq.~\eqref{Q}:
\bb
C_{d, \SIO}(\rho) = C_{d, \PIO}(\rho) = Q(\rho)\, .
\label{solution}
\ee
\end{thm}

\begin{proof}[Proof of Theorem~\ref{solution thm}]
Thanks to Lemma~\ref{C d SIO lower bound lemma}, it is only left to show that $C_{d,\SIO}(\rho)\leq Q(\rho)$. When $\lambda(\rho)=0$ and hence $\rho=\widebar{\rho}$ by Eq.~\eqref{lambda = 0}, we have that
\bbb
C_{d,\SIO}(\rho)\leq C_{d,\IO}(\rho) = C_r(\rho) = C_r(\widebar{\rho}) = Q(\rho)\, ,
\eee
where we employed the identities in Eq.~\eqref{results VV} and~\eqref{Q and Cr and Cf}.

The nontrivial case is thus when $\lambda(\rho)>0$. Let $r$ be an achievable rate for SIO distillation (in the sense of Eq.~\eqref{distillable}). Using Proposition~\ref{upper bound r prop} together with the upper bound in Proposition~\ref{casino prop} for $k=2^{\floor{rn}}$, we deduce that
\bb
\begin{aligned}
0 &\leq \lim_{\epsilon\to 0^+} \liminf_{n\to \infty} \bigg\{ \log \lambda(\rho) \\
&\quad + \log \left[ 1 + \frac{1}{\lambda(\rho)} \, 2^{\widetilde{H}_{\max}^{\epsilon^2\!/4} \left( J^n |(S^{\rho})^n \right)_{\delta_\rho^n} - \floor{rn} } \right] \bigg\} \\[1ex]
&= \log \lambda(\rho) \\
&\quad + \log \bigg[ 1 + \frac{1}{\lambda(\rho)} \, \lim_{\epsilon }  \, 2^{ \liminf_{n} \left( \widetilde{H}_{\max}^{\epsilon^2\!/4} \left( J^n |(S^{\rho})^n \right)_{\delta_\rho^n} - \floor{rn} \right)} \bigg] . \\[1ex] 
\end{aligned}
\label{solution proof eq1}
\ee
Now, since
\bbb
\lim_{n\to \infty} \frac1n \widetilde{H}_{\max}^{\epsilon^2/4} \left( J^n |(S^{\rho})^n \right)_{\delta_\rho^n} = H\left(J | S^\rho\right)_{\delta_\rho} = Q(\rho)
\eee
for all $\epsilon>0$ by Lemma~\ref{tweaked AEP lemma} and Eq.~\eqref{Q entropies}, we see that
\bbb
\liminf_{n\to \infty} \left( \widetilde{H}_{\max}^{\epsilon^2/4} \left( J^n |(S^{\rho})^n \right)_{\delta_\rho^n} - \floor{rn} \right) = -\infty
\eee
as soon as $r>Q(\rho)$. Thanks to Eq.~\eqref{solution proof eq1}, in this case we would obtain $\log \lambda(\rho)\geq 0$, absurd since $\lambda(\rho)<1$. Hence, we conclude that $r\leq Q(\rho)$, as claimed.
\end{proof}

We now explore some consequences of the above Theorem~\ref{solution thm}. The first problem we look into is that of the difference between the distillable coherence under SIO/PIO on one side and under IO on the other. It turns out that these two quantities coincide precisely for those states that are IO reversible according to~\cite[Theorem~10]{Winter2016}.

\begin{cor}
For a given state $\rho$,
\begin{align*}
\text{either}&\quad C_{d,\SIO/\PIO}(\rho)<C_{d,\IO}(\rho)<C_{c,\IO}(\rho) \\
\text{or}&\quad C_{d,\SIO/\PIO}(\rho) = C_{d,\IO}(\rho) = C_{c,\IO}(\rho)\, ,
\end{align*}
where IO denotes the set of incoherent operations. In other words, the IO reversible states of~\cite[Theorem~10]{Winter2016} are precisely those that can be distilled just as efficiently with PIO, SIO and IO.
\end{cor}

\begin{proof}
Remember that $C_{d,\IO}(\rho) = C_r(\rho)$ and $C_{c,\IO}(\rho)=C_f(\rho)$ by Eq.~\eqref{results VV}. We already observed in Remark~\ref{Q and Cr and Cf rem} that $C_r(\rho)=C_f(\rho)$ precisely when $\rho=\widebar{\rho}$, which is the same as requiring that $Q(\rho)=C_r(\rho)$.
\end{proof}

We now consider a slightly modified scenario for coherence distillation under SIO/PIO. Namely, drawing inspiration from~\cite{Thapliyal2003}, let us consider a class of incoherent operations $\pazocal{O}$ assisted by a sublinear amount of coherence. The corresponding distillation rate will be defined by
\bb
\begin{aligned}
&C_{d,\mathrm{\pazocal{O}c}}(\rho) \\
&\ \coloneqq \sup\Big\{ r : \lim_{n\to\infty} \inf_{\Lambda\in \pazocal{O}} \left\| \Lambda\big( \rho^{\otimes n}\! \otimes\! \Psi_{2^{o(n)}} \big)\! -\! \Psi_{2^{ \floor{rn}}} \right\|_1 = 0 \Big\} ,
\end{aligned}
\label{distillable sublinear}
\ee
where $o(n)$ is an arbitrary (positive and integer-valued) function of $n$ such that $\lim_{n\to\infty} o(n)/n=0$, over which we maximise the yield $r$.
\begin{rem}
The fact that the state we allow a sublinear amount of copies of in~\eqref{distillable sublinear} is a maximally coherent state is immaterial, at least for $\pazocal{O}\in\{\MIO,\DIO,\IO,\SIO\}$: since all states have finite coherence cost under either of these classes of operations, even under the constraint of zero error, we can use the maximally coherent states to distil a sublinear number of copies of any other state. The situation is not so clear for $\pazocal{O}=\PIO$, as some states there have infinite cost, and thus the corresponding scenarios may differ, at least in principle.
\end{rem}

In the case of entanglement theory under LOCC, it is known that a sublinear amount of Bell pairs does not help asymptotically at least in pure-to-pure state transformations~\cite{Thapliyal2003}. However, an analogous result for the mixed-state case is not known. This problem is connected with that of the full additivity of the distillable entanglement, which has been conjectured to be violated in general~\cite{Shor2001}. We lack a complete answer to this question even in the case where the free operations are only asymptotically non-entangling~\cite{BrandaoPlenio1, BrandaoPlenio2}. Thanks to Theorem~\ref{solution thm}, we can solve the corresponding problem in the theory of coherence completely.

\begin{cor}
For $\pazocal{O}\in\{\MIO,\DIO,\IO,\SIO,\PIO\}$, a sublinear amount of maximally coherent bits does not help distillation. In formula,
\bb
C_{d,\mathrm{\pazocal{O}c}}(\rho) \equiv C_{d,\pazocal{O}}(\rho)\, ,
\ee
for all states $\rho$, and for all classes of operations $\pazocal{O}\in\{\MIO,\DIO,\IO,\SIO,\PIO\}$.
\end{cor}

\begin{proof}
Clearly, it always holds that $C_{d,\mathrm{\pazocal{O}c}}(\rho)\geq C_{d,\pazocal{O}}(\rho)$. In order to show the converse, write
\begin{align*}
C_{d,\mathrm{\pazocal{O}c}}(\rho) &\textleq{1} \frac1n\, C_{d,\pazocal{O}}\left(\rho^{\otimes n} \otimes \Psi_{2^{\floor{\delta n}}}\right) \\
&\texteq{2} C_{d,\pazocal{O}}(\rho) + \frac1n \, C_{d,\pazocal{O}} \left( \Psi_{2^{\floor{\delta n}}}\right) \\
&= C_{d,\pazocal{O}}(\rho) + \frac1n \log \floor{\delta n} \\
&\leq C_{d,\pazocal{O}}(\rho) + \delta\, .
\end{align*}
Here, inequality 1 follows from the fact that for any $\delta>0$ we have that $\delta n\geq o(n)$ eventually in $n$. The equality in 2, instead, is an application of the full additivity of $C_{d,\pazocal{O}}$ for $\pazocal{O}\in\{\MIO,\DIO,\IO,\SIO,\PIO\}$, that is, of the identities
\bbb
C_{d,\pazocal{O}}(\omega\otimes \tau) = C_{d,\pazocal{O}}(\omega) + C_{d,\pazocal{O}}(\tau)\, ,
\eee
valid for all states $\omega,\tau$. In fact, by~\cite{Chitambar-reversible} we have that $C_{d,\MIO}=C_{d,\DIO}=C_{d,\IO}=C_r$, which is indeed fully additive~\cite{Winter2016}. Thanks to Theorem~\ref{solution thm}, we also know that $C_{d,\SIO}=C_{d,\PIO}=Q$, which is again fully additive by~\eqref{additivity Q}. Since we have established that $C_{d,\mathrm{\pazocal{O}c}}(\rho)\leq C_{d,\pazocal{O}}(\rho)+\delta$ for all $\delta>0$, it must be the case that $C_{d,\mathrm{\pazocal{O}c}}(\rho)\leq C_{d,\pazocal{O}}(\rho)$, which completes the proof.
\end{proof}

\section{Coherence cost under PIO} \label{sec PIO cost}

Throughout Sections~\ref{sec direct}--\ref{sec converse} we have shown that SIO/PIO are overall weak sets of operations as far as coherence distillation is concerned. Intuitively, we can see this as a consequence of the unavoidable noise such operations introduce into the system. When the aim is to prepare maximally coherent states, even the slightest amount of noise will be detrimental to the process. In turn, this entails that with SIO/PIO one cannot do much more than `isolate' the coherence that was already there in the system, while there is no hope to `concentrate' it if it was dispersed in the first place.

However, this picture changes dramatically when we consider the task of coherence dilution instead of that of coherence distillation. As we discussed in Section~\ref{subsec PIO cost}, in coherence dilution we aim to prepare $n$ copies of a target state $\rho$ by means of operations in a certain class and using up as resources some $\floor{rn}$ coherence bits $\Psi_2$. Although we allow for a small error in this preparation process, we require that such error vanishes in the asymptotic limit $n\to\infty$. The maximal rate $r$ for which this procedure is possible is known was defined in Eq.~\eqref{cost} as the cost of the state $\rho$ relative to the given class of operations. As the name suggests, in coherence dilution noisy operations are not necessarily useless, as the target state can be noisy itself. At the intuitive level, this may lead us to conjecture that sufficiently mixed states have finite SIO/PIO cost.

Indeed, this turns out to be the case. The problem of coherence dilution under SIO has been solved in~\cite{Xiao2015, Winter2016}, where it was shown that 
\bb
C_{c,\SIO}(\rho) = C_{c,\IO}(\rho) = C_f(\rho)\, ,
\label{distillable SIO/IO}
\ee
where $C_f$ is the coherence of formation defined in Eq.~\eqref{coherence of formation}. 
To prove Eq.~\eqref{distillable SIO/IO}, observe that: (i) by~\cite[Theorem~8]{Winter2016}, the IO coherence cost is given by $C_{c,\IO}(\rho) = C_f(\rho)$; (ii) since SIO are special cases of IO, it holds that $C_{c,\SIO}(\rho)\geq C_{c,\IO}(\rho)$ by construction; (iii) however, $C_{c,\SIO}(\rho)\leq C_f(\rho)$, because the protocol described in the proof of~\cite[Theorem~8]{Winter2016} involves only the preparation of pure states, and SIO are equivalent to IO as far as pure-to-pure transformation are concerned~\cite{Xiao2015}.

The above construction confirms our intuition: although much weaker than IO at distilling coherence bits, SIO are as powerful as IO when it comes to coherence dilution. The relevance of the above results for experimental practice is however hindered by the fact that the implementation of SIO still requires coherent (destructive) measurements on ancillary systems~\cite{Yadin2016}. To obtain a feasible protocol to prepare states in a reliable way with minimal coherence consumption, we instead need to look at physically incoherent operations, i.e.\ PIO. As we mentioned before, the problem of computing the PIO coherence cost does not seem to have been considered before. In this section we solve this problem completely by providing an analytical formula for the PIO coherence cost of a generic state (Theorem~\ref{C c PIO thm}). Curiously, this turns out to be given by an expression similar to that of the coherence of formation (Eq.~\eqref{coherence of formation}), but with the infimum running over convex decompositions comprising uniformly coherent states only. We duly dub this quantity \emph{uniform coherence of formation}. After discussing some preliminary notions in Subsection~\ref{subsec preliminaries}, we introduce and study the uniform coherence of formation in Section~\ref{subsec CfM}. In the subsequent~\eqref{subsec PIO cost final} we finally give the full proof of Theorem~\ref{C c PIO thm}.

\subsection{Preliminaries} \label{subsec preliminaries}

Remember that we defined $\M$ as the set formed by all uniformly coherent states of sizes $k=1,\ldots, d$. In what follows we will look at its convex hull, denoted $\co (\M)$. Observe that since $\M$ is compact the same is true of $\co(\M)$. The reason of our interest lies in the following fact.

\begin{lemma} \label{PIO preserve co(M) lemma}
Physically incoherent operations preserve the set $\co(\M)$. Namely, if $\rho\in \co(\M)$ and $\Lambda$ is a PIO then also $\Lambda(\rho)\in \co(\M)$.
\end{lemma}

\begin{proof}
Up to convex combinations, it suffices to show the claim for an elementary PIO $\Lambda$ acting as in~\eqref{elementary PIO Kraus}. Given a decomposition of $\rho\in\co(\M)$ as $\rho=\sum_\alpha p_\alpha \Psi_\alpha$, with $\ket{\Psi_\alpha}\in\M_{k_\alpha}$, we write
\bbb
\Lambda(\rho) = \sum_{\alpha,\beta} p_\alpha U_\beta \Pi_\beta \Psi_\alpha \Pi_\beta U_\beta^\dag \in \co(\M)\, ,
\eee
where we observed that $U_\beta \Pi_\beta \ket{\Psi_\alpha}$ is proportional to a uniformly coherent state for all $\alpha,\beta$.
\end{proof}

The above result immediately shows that since in the task of PIO coherence dilution we start from a maximally coherent (and hence uniformly coherent) state, we cannot hope to construct any state that does not belong to $\co(\M)$. Since the extreme points of $\co(\M)$ are precisely the uniformly coherent states, any other pure state necessarily lies outside of it. Thus, we immediately retrieve the result of~\cite{Chitambar-PIO} that all pure states except for the uniformly coherent ones have infinite PIO coherence cost. We now take the chance to extend this result by providing a simple necessary criterion to check whether a given state is in $\co(\M)$ or not.

\begin{lemma} \label{excluded states lemma}
Every state $\rho\in\co(\M)$ has the property that $|\rho_{ij}|\leq \min\{\rho_{ii}, \rho_{jj}\}$ for all pairs of indices $i,j$. Consequently, the only pure states in $\co(\M)$ are uniformly coherent.
\end{lemma}

\begin{proof}
Let $\rho$ be decomposed as $\rho=\sum_\alpha p_\alpha \Psi_\alpha$, where the states $\Psi_\alpha = \ketbra{\Psi_\alpha}$ are uniformly coherent on a set $J_\alpha$. For $i\in [d]$ and a subset $J\subseteq [d]$, define the Kronecker delta symbol as $\delta_{i,J} =1$ if $i\in J$ and $0$ otherwise. Then for all $\alpha$
\bbb
\left| (\Psi_\alpha)_{ij} \right| = \frac{\delta_{i,J_\alpha} \delta_{j,J_\alpha}}{k_\alpha} = \min\left\{ (\Psi_\alpha)_{ii},\, (\Psi_\alpha)_{jj}\right\} ,
\eee
implying that
\begin{align*}
|\rho_{ij}| &= \left| \sumno_\alpha p_\alpha (\Psi_\alpha)_{ij} \right| \\[.5ex]
&\leq \sum_\alpha p_\alpha \left| (\Psi_\alpha)_{ij}\right| \\
&= \sum_\alpha p_\alpha \min\left\{ (\Psi_\alpha)_{ii},\, (\Psi_\alpha)_{jj}\right\} \\
&\leq \min\left\{ \sumno_\alpha p_\alpha (\Psi_\alpha)_{ii},\, \sumno_\alpha p_\alpha (\Psi_\alpha)_{jj} \right\} \\
&= \min\left\{\rho_{ii},\, \rho_{jj}\right\} .
\end{align*}
This proves the first claim. Now, take a pure state $\ket{\psi} = \sum_i z_i \ket{i}$. If $\ket{\psi}$ is not uniformly coherent then $0<|z_i|<|z_j|$ for some $i,j$. It is then easy to see that the projector $\psi=\ketbra{\psi}$ satisfies $|\psi_{ij}| = |z_i| |z_j| > |z_i|^2 = \psi_{ii}$, implying that $\psi\notin\co(\M)$.
\end{proof}

The above result could make us fear that the set $\co(\M)$ is too meagre, which would seriously hinder PIO coherence dilution via Lemma~\ref{PIO preserve co(M) lemma}. Fortunately, we now show that this is not the case. In fact, although $\M$ has measure zero, $\co(\M)$ turns out to have nonzero volume, as can be proved e.g.\ by showing that it contains a full ball around the maximally mixed state.

\begin{lemma} \label{diag dom state conv M lemma}
Let $\rho$ be a diagonally dominant state on a system of dimension $d$, i.e.\ let it be such that $\rho_{ii} \geq \sum_{j\neq i} |\rho_{ij}|$ for all $i=1,\ldots,d$. Then $\rho\in\co(\M)$. Consequently, for every state $\rho$ in dimension $d$ it holds that
\bb
\left\|\rho - \frac{\id}{d}\right\|_{1 \to 1} \leq \frac1d\quad \Longrightarrow\quad \rho \in  \co(\M)\, ,
\label{sufficient conv M}
\ee
where $\|X\|_{1 \to 1}\coloneqq \max_{1\leq i\leq d} \sum_{j=1}^d |X_{ij}|$ is the so-called max-row sum norm.
\end{lemma}

\begin{rem}
Using the Cauchy--Schwartz inequality, it is not difficult to show that all states with low enough purity automatically obey Eq.~\eqref{sufficient conv M}. Namely, if $\Tr \rho^2\leq \frac{d^2+1}{d^3}$ for a state in dimension $d$ then Eq.~\eqref{sufficient conv M} is necessarily satisfied.
\end{rem}

\begin{proof}[Proof of Lemma~\ref{diag dom state conv M lemma}]
Mimicking a technique first employed in~\cite{Ringbauer2017, Johnston18}, we can write a diagonally dominant state $\rho$ as a convex combination
\begin{align*}
\rho &= \sum_{i<j} |\rho_{ij}| \, \frac{\frac{\rho_{ij}}{|\rho_{ij}|} \ket{i} + \ket{j}}{\sqrt2} \, \frac{\frac{\rho_{ij}^*}{|\rho_{ij}|} \bra{i} + \bra{j}}{\sqrt2} \\
&\qquad + \sum_{i} \left(\rho_{ii} - \sumno_{j\neq i} |\rho_{ij}|\right) \ketbra{i}\, ,
\end{align*}
which shows that $\rho\in \co\left( \M_1 \cup \M_2\right)\subseteq \co(\M)$ and proves the first claim. Finally, the second claim follows from the elementary observation that if $\|X\|_{1\to 1}\leq 1$ then $\id+X$ is diagonally dominant.
\end{proof}

\subsection{Uniform coherence of formation} \label{subsec CfM}

Given a state $\rho\in\M$, we define its \emph{uniform coherence of formation} as the convex roof
\bb
C_{f}^\M(\rho) \coloneqq \inf_{\substack{\\ \sum_\alpha p_\alpha \Psi_\alpha = \rho \\[.5ex] \ket{\Psi_\alpha}\in\M_{k_\alpha}}} \sum_\alpha p_\alpha \log k_\alpha\, ,
\label{CfM}
\ee
where $\Psi_\alpha = \ketbra{\Psi_\alpha}$, as usual. If the infimum in~\eqref{CfM} is over an empty set, that is, if $\rho\notin \co(\M)$, we set by convention $C_f^\M(\rho)\coloneqq +\infty$. Observe that for all uniformly coherent states $\ket{\Psi}\in\M_k$ it holds that $C_f^\M(\Psi) = \log k$. As is easy to see, the proof of Lemma~\ref{diag dom state conv M lemma} shows that all diagonally dominant states $\rho$, thereby including those obeying the inequality in Eq.~\eqref{sufficient conv M}, satisfy $C_f^\M(\rho)\leq 1$. On a different line, it follows e.g.\ from~\cite[Lemma~A.2]{Uhlmann1998} that the infimum in Eq.~\eqref{CfM} is always achieved on a decomposition formed by no more than $d^2$ elements, with $d$ being the dimension of the underlying Hilbert space.

\begin{prop} \label{properties CfM prop}
The uniform coherence of formation is:
\begin{enumerate}[(a)]
\item a convex strong monotone under PIO, meaning that for all states and for all PIO $\Lambda$ acting as in~\eqref{PIO Kraus} it holds that
\bb
\begin{aligned}
C_f^\M\left( \Lambda(\rho)\right) &\leq \sum_{\alpha,\beta} p_\alpha\, q_{\beta|\alpha}\, C_f^\M\big( \widetilde{\rho}_{\beta |\alpha} \big) \\
&\leq C_f^\M(\rho)\, ,
\end{aligned}
\label{strong monotonicity CfM}
\ee
where $q_{\beta|\alpha}\coloneqq \Tr\big[\rho\, \Pi_{\beta|\alpha} \big]$ and $\widetilde{\rho}_{\beta|\alpha}\coloneqq \frac{1}{q_{\beta|\alpha}}\, U_{\alpha,\beta} \Pi_{\beta|\alpha}\, \rho \, \Pi_{\beta|\alpha} U_{\alpha,\beta}^\dag$.
\item superadditive, i.e. such that
\bb
C_f^\M (\rho_{AB}) \geq C_f^\M(\rho_A) + C_f^\M(\rho_B)
\label{superadditivity CfM}
\ee
for all bipartite states $\rho_{AB}$, where $\rho_A\coloneqq \Tr_B\rho_{AB}$ and similarly for $\rho_B$;
\item fully additive on tensor products, meaning that
\bb
C_f^\M (\rho_A \otimes \sigma_B) = C_f^\M(\rho_A) + C_f^\M(\sigma_B)
\label{additivity CfM}
\ee
for all pairs of states $\rho_A$, $\sigma_B$;
\item lower semicontinuous.
\end{enumerate}
\end{prop}

\begin{proof}
We start from claim (a). The fact that $C_f^\M$ is convex follows immediately from its definition. To show that it is also a strong PIO monotone, it then suffices to consider an elementary PIO $\Lambda$ that acts as in~\eqref{elementary PIO Kraus}. Consider the decomposition $\rho=\sum_\gamma p_\gamma \Psi_\gamma$ of $\rho$ that achieves the infimum in Eq.~\eqref{CfM}, in formula $C_f^\M(\rho) = \sum_\gamma p_\gamma \log k_\gamma$, where $\ket{\Psi_\gamma} \in \M_{k_\gamma}$. Clearly, for all $\beta$ and $\gamma$ we will have that $U_\beta \Pi_\beta \ket{\Psi_\gamma}$ is proportional to some uniformly coherent state of size $k_{\beta|\gamma}$:
\bbb
U_\beta \Pi_\beta \ket{\Psi_\gamma} = \sqrt{\frac{k_{\beta|\gamma}}{k_\gamma}} \ket{\Psi_{\beta|\gamma}},\qquad \ket{\Psi_{\beta|\gamma}} \in \M_{k_{\beta|\gamma}}\, .
\eee
Moreover, the completeness relation $\sum_\beta \Pi_\beta=\id$ imposes that
\bb
\sumno_\beta k_{\beta|\gamma} = k_\gamma\qquad \forall\ \gamma\, .
\label{sum k alpha beta}
\ee
The post-measurement state corresponding to outcome $\beta$ is $\widetilde{\rho}_\beta = \sumno_\gamma \frac{p_\gamma k_{\beta|\gamma}}{q_\beta k_\gamma}\, \Psi_{\beta|\gamma}$, where $q_\beta = \sum_{\gamma'} \frac{p_{\gamma'} k_{\beta|\gamma'}}{k_{\gamma'}}$. 
Note that $\Lambda(\rho)= \sum_\beta q_\beta \widetilde{\rho}_{\beta}$. Using the convexity of $C_f^\M$, the concavity of the logarithm, and the normalisation condition~\eqref{sum k alpha beta}, we deduce that
\begin{align*}
C_f^\M\left(\Lambda(\rho)\right) &\leq \sum_\beta q_\beta C_f^\M \left(\widetilde{\rho}_\beta\right) \\
&\leq \sum_\beta q_\beta \sum_\gamma \frac{p_\gamma k_{\beta|\gamma}}{q_\beta k_\gamma}\, \log k_{\beta|\gamma} \\
&= \sum_\gamma p_\gamma \sum_\beta \frac{k_{\beta|\gamma}}{k_\gamma}\, \log k_{\beta|\gamma} \\
&\leq \sum_\gamma p_\gamma \log \left( \frac{1}{k_\gamma} \sumno_\beta k_{\beta|\gamma}^2\right) \\
&\leq \sum_\gamma p_\gamma \log \left( \frac{1}{k_\gamma} \left(\sumno_\beta k_{\beta|\gamma}\right)^2 \right) \\
&= \sum_\gamma p_\gamma \log k_\gamma \\
&= C_f^\M(\rho)\, .
\end{align*}
This completes the proof of~\eqref{strong monotonicity CfM} and of claim (a).

To prove claim (b), it suffices to show that for all uniformly coherent bipartite pure states $\ket{\Psi}_{AB}\in\M_k$ one has that
\bb
\log k = C_f^\M (\Psi_{AB}) \geq C_f^\M(\Psi_A) + C_f^\M(\Psi_B)\, .
\label{superadditity CfM pure states}
\ee
Before delving into the proof of Eq.~\eqref{superadditity CfM pure states}, let us show how this allows us to deduce Eq.~\eqref{superadditivity CfM}. For a generic decomposition $\rho_{AB}=\sum_\alpha p_\alpha \Psi_{AB}^{(\alpha)}$ of $\rho_{AB}$ into uniformly coherent states $\ket{\Psi^{(\alpha)}}_{AB}\in\M_{k_\alpha}$, using Eq.~\eqref{superadditity CfM pure states} we would obtain that
\begin{align*}
\sum_\alpha p_\alpha \log k_\alpha &= \sum_\alpha p_\alpha C_f^\M \left(\Psi_{AB}^{(\alpha)} \right) \\
&\geq \sum_\alpha p_\alpha \left( C_f^\M \left(\Psi_{A}^{(\alpha)} \right) + C_f^\M \left(\Psi_{B}^{(\alpha)} \right) \right) \\
&= \sum_\alpha p_\alpha C_f^\M \left(\Psi_{A}^{(\alpha)} \right) + \sum_\alpha p_\alpha C_f^\M \left(\Psi_{B}^{(\alpha)} \right) \\
&\geq C_f^\M(\Psi_A) + C_f^\M \left(\Psi_{B} \right) ,
\end{align*}
where the last inequality follows from the convexity of $C_f^\M$. Since the decomposition of $\rho_{AB}$ we considered was entirely arbitrary, this would imply Eq.~\eqref{superadditivity CfM}.

We now prove Eq.~\eqref{superadditity CfM pure states}. Write
\bb
\ket{\Psi}_{AB} = \frac{1}{\sqrt{k}}\sum_{i,j} M_{ij} \ket{ij} \in \M_k\, , 
\ee
where each entry of the $d_A\times d_B$ complex matrix $M$ is either $0$ or has modulus $1$, in formula $|M_{ij}|\in \{0,1\}$ for all $i,j$. Define
\bb
k_i \coloneqq \sum_j |M_{ij}|\, ,\qquad h_j \coloneqq \sum_i |M_{ij}|\, ,
\label{ki hj}
\ee
so that 
\bb
\sum_i k_i = \sum_j h_j = k\, .
\label{ki hj normalisation}
\ee
A quick calculation reveals that the partial trace $\Psi_A = \Tr_B \Psi_{AB}$ takes the form
\bbb
\Psi_A = \sum_j \frac{h_j}{k} \left( \frac{1}{\sqrt{h_j}} \sumno_i M_{ij}\ket{i}\right) \left( \frac{1}{\sqrt{h_j}} \sumno_i M_{ij}^* \bra{i}\right) .
\eee
Since this is a convex decomposition of $\Psi_A$ into uniformly coherent states, we deduce the estimate
\bbb
C_f^\M\left(\Psi_A\right) \leq \sum_j \frac{h_j}{k} \log h_j \, ;
\eee
analogously, it can be shown that
\bbb
C_f^\M\left(\Psi_B\right) \leq \sum_i \frac{k_i}{k} \log k_i \, .
\eee
Putting all together yields
\begin{align*}
C_f^\M\left(\Psi_A\right) + C_f^\M\left(\Psi_B\right) &\leq \sum_j \frac{h_j}{k} \log h_j + \sum_i \frac{k_i}{k} \log k_i \\
&\texteq{1} \sum_{i,j} \frac{|M_{ij}|}{k} \log h_j + \sum_{i,j} \frac{|M_{ij}|}{k} \log k_i \\
&= \sum_{i,j} \frac{|M_{ij}|}{k} \log (k_i h_j) \\
&\textleq{2} \log \left( \sumno_{i,j} \frac{|M_{ij}|}{k}\, k_i h_j \right) \\
&\textleq{3} \log \left( \frac1k \sumno_{i,j} k_i h_j \right) \\
&= \log \left( \frac1k \left(\sumno_i k_i\right) \left( \sumno_j h_j\right) \right) \\
&\texteq{4} \log k \, .
\end{align*}
The justification of the above derivation is as follows. 1: Comes from Eq.~\eqref{ki hj}; 2: is a consequence of the concavity of the logarithm, once one observes $|M_{ij}|/k$ is a probability distribution over $[d_A d_B]$ by Eq.~\eqref{ki hj normalisation}; 3: is an application of the inequality $|M_{ij}|\leq 1$; finally, 4: descends once again from the normalisation condition in Eq.~\eqref{ki hj normalisation}. This completes the proof of claim (b).

We now move on to (c). Applying (b) in the special case where $\rho_{AB}$ is a product state we arrive at the inequality $C_f^\M(\rho_A\otimes \sigma_B) \geq C_f^\M(\rho_A)+C_f^\M(\sigma_B)$. The converse relation is easily established by taking two decompositions $\rho_A=\sum_\alpha p_\alpha \Psi_\alpha$ and $\sigma_B = \sum_\beta q_\beta \Psi'_\beta$ that achieve the infima defining the uniform coherences of formation of $\rho_A$ and $\sigma_B$, respectively, and considering the `product' decomposition $\rho_A\otimes \sigma_B = \sum_{\alpha,\beta} p_\alpha q_\beta \Psi_\alpha\otimes \Psi'_\beta$.

To establish (d), we have to show that for all sequences of states $\rho_n$ such that $\lim_{n\to\infty}\rho_n = \rho$ it holds that $C_f^\M(\rho)\leq \liminf_{n\to\infty} C_f^\M(\rho_n)$. Up to taking subsequences, we can assume that the $\liminf$ is in fact a proper limit. Also, since if $\rho_n\notin \co(\M)$ frequently in $n$ then $\lim_{n\to\infty} C_f^\M(\rho_n) = \infty$, we can assume that we are in the nontrivial case where $\rho_n\in \co(\M)$ eventually.
Remembering that in dimension $d$ the optimal decomposition in Eq.~\eqref{CfM} can be taken to be formed by no more than $d^2$ elements, we can write $\rho_n =\sum_{\alpha=1}^{d^2} p_\alpha^{(n)} \Psi_\alpha^{(n)}$ for some uniformly coherent vectors $\ket{\Psi_\alpha^{(n)}}\in \M_{k_\alpha^{(n)}}$ such that
\bbb
C_f^\M(\rho_n) = \sum_{\alpha=1}^{d^2} p_\alpha^{(n)} \log k_\alpha^{(n)}\, .
\eee
By compactness, up to picking a further subsequence we can assume that $\lim_{n\to\infty} p_\alpha^{(n)}\eqqcolon p_\alpha$ and $\lim_{n\to\infty} \ket{\Psi_\alpha^{(n)}}\eqqcolon \ket{\Psi_\alpha}\in \M_{k_\alpha}$ exist for all $\alpha=1,\ldots, d^2$. Clearly, we will also have $\lim_{n\to\infty} k_\alpha^{(n)}=k_\alpha$, the sequence on the r.h.s.\ being actually eventually constant. Taking the limit on that subsequence we see that $\rho=\lim_{n\to\infty}\rho_n = \sum_{\alpha=1}^{d^2} p_\alpha \Psi_\alpha$ is in fact a legitimate decomposition of $\rho$, from which it follows that
\begin{align*}
C_f^\M(\rho) &\leq \sum_{\alpha=1}^{d^2} p_\alpha \log k_\alpha \\
&= \lim_{n\to\infty} \sum_{\alpha=1}^{d^2} p_\alpha^{(n)} \log k_\alpha^{(n)} \\
&= \lim_{n\to\infty} C_f^\M(\rho_n)\, ,
\end{align*}
completing the proof.
\end{proof}

\begin{rem}
The superadditivity of the standard coherence of formation was recently established in~\cite{Liu2018}, where a set of general conditions to determine superadditivity of convex roof coherence measures was also studied. However, due to the constrained nature of the convex roof that defines the uniform coherence of formation, the result in~\cite{Liu2018} does not appear to be directly applicable here.
\end{rem}

\begin{rem}
We suspect that $C_f^\M$ is in fact continuous and even asymptotically continuous when finite, i.e.\ on $\co(\M)$. However, we could not decide whether this is actually the case. In what follows we will only need the properties guaranteed by the above Proposition~\ref{properties CfM prop}.
\end{rem}

Before we close this section, we allow ourselves to present a simple yet general lower bound that -- among other things -- enables us to compute the uniform coherence of formation for all qubit states.

\begin{prop} \label{CfM qubits prop}
For all states $\rho$ it holds that
\bb
C_f^\M (\rho) \geq 2 \max_{i\neq j} |\rho_{ij}|\, .
\label{lower bound CfM}
\ee
In particular, for all single-qubit states
\bb
\rho = \begin{pmatrix} p & z \\ z^* & 1-p \end{pmatrix}
\label{qubit state}
\ee
the uniform coherence of formation of Eq.~\eqref{CfM} evaluates to
\bb
C_f^\M(\rho) = \left\{ \begin{array}{ll} 2|z| & \text{if $|z|\leq \min\{p,1-p\}$,} \\[1ex] +\infty & \text{otherwise.} \end{array}\right.
\label{CfM qubits}
\ee
\end{prop}

\begin{proof}
To prove~\eqref{lower bound CfM}, take an arbitrary decomposition $\rho = \sum_\alpha p_\alpha \Psi_\alpha$, where $\ket{\Psi_\alpha}\in \M_{k_\alpha}$ for all $\alpha$. Consider two indices $i\neq j$, and write
\begin{align*}
|\rho_{ij}| &= \left| \sumno_\alpha p_\alpha (\Psi_\alpha)_{ij} \right| \\
&\leq \sumno_\alpha p_\alpha \left| (\Psi_\alpha)_{ij} \right| \\
&\texteq{1} \sumno_{\alpha:\, k_\alpha\geq 2} p_\alpha \left| (\Psi_\alpha)_{ij} \right| \\
&\textleq{2} \sumno_{\alpha:\, k_\alpha\geq 2} p_\alpha \frac{1}{k_\alpha} \\
&\textleq{3} \frac12 \sumno_{\alpha:\, k_\alpha\geq 2} p_\alpha \log k_\alpha \\
&= \frac12 \sumno_\alpha p_\alpha \log k_\alpha\, .
\end{align*}
Here, in step 1 we observed that only terms with $k_\alpha\geq 2$ contribute to the sum, for step 2 we noted that those terms that do contribute do so with a coefficient $1/k_\alpha$, and in step 3 we applied the elementary inequality $1/k\leq (\log k)/2$, valid for all $k\geq 2$. Taking the infimum over all decompositions of $\rho$ and the maximum over all pairs $i\neq j$ yields the sought lower bound.

From~\eqref{lower bound CfM} it follows immediately that $C_f^\M(\rho)\geq 2|z|$ for all qubit states parametrised as in~\eqref{qubit state}. To show that this lower bound is in fact tight, start by noting that if $|z| > \min\{p,1-p\}$ then by Lemma~\ref{excluded states lemma} we have that $\rho\notin \co(\M)$, i.e.\ $C_f^\M(\rho)=+\infty$. Otherwise, proceeding as in the proof of Lemma~\ref{diag dom state conv M lemma} we can write the decomposition
\bb
\rho = 2|z| \ketbra{\Psi_z} + \left( p - |z|\right) \ketbra{0} + \left(1-p - |z|\right) \ketbra{1}\, ,
\ee
where $\ket{\Psi_z}\coloneqq \frac{1}{\sqrt2} \left( \ket{0} + \frac{z^*}{|z|} \ket{1}\right)$, from which it follows that $C_f^\M(\rho) \leq 2|z|$. This completes the proof of~\eqref{CfM qubits}.
\end{proof}

\subsection{Coherence cost under PIO} \label{subsec PIO cost final}

We are finally ready to prove the main result of this section.

\setthmtag{\ref{C c PIO thm}}
\begin{thm}
The PIO coherence cost is given by the uniform coherence of formation:
\bb
C_{c,\PIO} (\rho) = C_f^\M(\rho)
\ee
for all states $\rho$.
\end{thm}

\begin{proof}[Proof of Theorem~\ref{C c PIO thm}]
The argument is as usual composed of two parts: first we prove the achievability part (direct statement), then we show that what obtained is in fact the optimal rate (converse). While the direct statement follows some pretty standard arguments similar e.g.\ to those in~\cite{Hayden-EC}, establishing the converse is less straightforward. The reason of this is that we lack an indispensable tool to pursue the standard strategy, i.e.\ the asymptotic continuity of the uniform coherence of formation. Fortunately, we will see that superadditivity and lower semicontinuity as established by Proposition~\ref{properties CfM prop} can serve the purpose just as well. As far as we know, this rather peculiar proof strategy has not been used before in quantum information theory.

We start by proving the direct statement: the uniform coherence of formation is an achievable rate for the dilution process. This part of the proof mimics similar standard arguments to show e.g.\ that the entanglement of formation is an upper bound on the entanglement cost. When $\rho\notin \co(\M)$ then $C_f^\M(\rho) = +\infty$ and there is nothing to prove. We will therefore assume that $\rho\in\co(\M)$. Our goal is to show that for all $0<\delta<1$ and for all decompositions $\rho = \sum_{\alpha=1}^{d^2} p_\alpha \Psi_\alpha$ with $\ket{\Psi_\alpha}\in\M_{k_\alpha}$ the number $r = \sum_{\alpha=1}^{d^2} p_\alpha \log k_\alpha + \delta$ is an achievable rate. 

We start by drawing $n$ independent instances $\alpha^n=(\alpha_1,\ldots, \alpha_n)$ of a discrete random variable whose probability distribution is $p$. By the law of large numbers, with probability $P_{n,\delta'}$ approaching $1$ as $n\to\infty$, each symbol $\alpha\in [d^2]$ will appear in the sequence $\alpha^n$ no more than $n(p_\alpha + \delta')$ times, where $0<\delta'<1$ will be fixed later. Since we admit an asymptotically vanishing error, we assume that the sequence $\alpha^n$ satisfies this property, which we signify by calling it \emph{strongly typical}~\cite[Definition~14.7.2]{MARK}. We now construct the sequence of pure states $\Psi_{\alpha_1}, \ldots, \Psi_{\alpha_n}$ allowing for a small error; this can be done by first generating $\sigma(\alpha^n) \coloneqq \bigotimes_{\alpha=1}^{d^2} \Psi_\alpha^{\otimes \floor{n (p_\alpha + \delta')}}$, and then rearranging or discarding the subsystems (which are allowed operations in the PIO setting). Thanks to Lemma~\ref{PIO cost pure states lemma} we know that a state $\omega(\alpha|\alpha^n) \approx_{\epsilon_\alpha(n)} \Psi_\alpha^{\otimes \floor{n (p_\alpha + \delta')}}$ can be obtained via PIO by consuming no more than $n (p_\alpha + \delta') \left( \log k_\alpha +\delta'\right)$ coherence bits. The approximation error $\epsilon_\alpha(n)$ satisfies $\lim_{n\to\infty} \epsilon_\alpha(n)=0$ for all $\alpha\in [d^2]$. Observe that as long as $\alpha^n$ is strongly typical, $\epsilon_\alpha(n)$ depends only on the symbol $\alpha\in [d^2]$ and not on the whole sequence $\alpha^n$. Setting $\epsilon(n)\coloneqq \sum_{\alpha=1}^{d^2} \epsilon_\alpha(n)$, thanks to the fact that $\alpha$ has a finite range we also get that $\lim_{n\to\infty} \epsilon(n)=0$. 

By rearranging and discarding subsystems we can now go from the state $\omega(\alpha^n)\coloneqq \bigotimes_{\alpha=1}^{d^2} \omega(\alpha|\alpha^n) \approx_{\epsilon(n)} \sigma(\alpha^n)$ to some $\tau(\alpha^n) \approx_{\epsilon(n)} \Psi_{\alpha_1}\otimes\ldots\otimes \Psi_{\alpha_n}$. The total cost of this protocol is upper bounded by
\begin{align*}
&n \sum_{\alpha=1}^{d^2} (p_\alpha + \delta') \left( \log k_\alpha +\delta'\right) \\
&\qquad \leq n \left( \sumno_{\alpha=1}^{d^2} p_\alpha \log k_\alpha + \delta' \left( d^2 (\log d+1) + 1\right) \right) \\
&\qquad \leq n \left( \sumno_{\alpha=1}^{d^2} p_\alpha \log k_\alpha + \delta \right) \\
&\qquad = n r\, ,
\end{align*}
where in the last line we picked $\delta'\coloneqq \delta/(d^2(\log d+1)+1)$ sufficiently small as a function of $\delta$. Forgetting the sequence $\alpha^n$, which is however known to be strongly typical, we obtain a state
\begin{align*}
\omega_{n,\delta'} &\coloneqq\ \sum_{\alpha^n\in \pazocal{T}_{n,\delta'}} p^{n}(\alpha^n |\pazocal{T}_{n,\delta'}) \, \omega(\alpha^n) \\
&\,\approx_{\epsilon(n)} \sum_{\alpha^n\in \pazocal{T}_{n,\delta'}} p^{n}(\alpha^n |\pazocal{T}_{n,\delta'})\, \Psi_{\alpha_1} \otimes \ldots \otimes \Psi_{\alpha_n} \\
&\, \eqqcolon \tau_{n,\delta'}\, ,
\end{align*}
where we denoted with $p^{n}(\alpha^n |\pazocal{T}_{n,\delta'})$ the conditional probability distribution of $\alpha^n$ in the strongly typical set $\pazocal{T}_{n,\delta'}$. Since this set contains asymptotically almost all the probability, i.e.\ $\lim_{n\to\infty} p^{n}\left( \pazocal{T}_{n,\delta'}\right) = 1$ for all $\delta'>0$~\cite[Property~14.7.2]{MARK}, it is not difficult to verify that
\bbb
\tau_{n,\delta'} \approx_{\epsilon'(n)} \sum_{\alpha^n} p^{n}(\alpha^n) \Psi_{\alpha_1} \otimes \ldots \otimes \Psi_{\alpha_n} = \rho^{\otimes n}
\eee
where $\lim_{n\to\infty} \epsilon'(n) =0$. This concludes the proof of the direct statement.

As we discussed above, the converse makes heavy use of the the properties of the uniform coherence of formation we established in Proposition~\ref{properties CfM prop}. We have to show that $C_{c,\PIO}(\rho) \geq C_f^\M(\rho)$ for all states $\rho$. First of all, if $\rho\notin \co(\M)$ and thus $C_f^\M (\rho)=+\infty$ then $\rho$ lies at a nonzero distance from the compact set $\co(\M)$. Since by applying PIO to a maximally coherent state one cannot go outside of $\co(\M)$ by Lemma~\ref{PIO preserve co(M) lemma}, even formation of a single copy of $\rho$ with vanishing error is impossible with PIO in this case. Hence, $C_{c,\PIO}(\rho)=+\infty$, confirming the inequality.

From now on we shall therefore assume that $\rho\in \co(\M)$ and thus $C_f^\M(\rho)\leq \log d$. Let $(\Lambda_n)_{n\in\N}$ be a sequence of PIO protocols such that the output states $\sigma_n = \sigma_n^{A_1\ldots A_n} \coloneqq \Lambda_n \left(\Psi_{2^{\floor{rn}}}\right)$ satisfy
\bbb
\lim_n \left\| \sigma_n^{A_1\ldots A_n} - \bigotimes\nolimits_{i=1}^n \rho^{A_i}\right\|_1=0\, ,
\eee
where we denoted with $A_1,\ldots, A_n$ the output systems. This amounts to saying that $r$ is an achievable rate for the formation of $\rho$ under PIO. Calling $\sigma_n^{(i)}$ the reduced state of $\sigma_n$ on the subsystem $A_i$, we now write
\begin{align*}
\floor{rn} &\texteq{1} C_f^\M \left( \Psi_{2^{\floor{rn}}} \right) \\
&\textgeq{2} C_f^\M\left( \sigma_n \right) \\
&\textgeq{3} \sum_{i=1}^n C_f^\M \Big( \sigma_n^{(i)} \Big) ,
\end{align*}
where step 1 follows from because $C_f^\M(\Psi_k) =\log k$ for all uniformly coherent states $\Psi_k$ of size $k$, step 2 comes from the monotonicity of $C_f^\M$ under PIO (Proposition~\ref{properties CfM prop}(a)), and finally step 3 derives from its superadditivity (Proposition~\ref{properties CfM prop}(b)).

The above inequality tells us that for all $n$ one can pick an index $1\leq i_n\leq n$ such that $\omega_n\coloneqq \sigma_n^{(i_n)}$ satisfies
\bbb
C_f^\M (\omega_n) \leq \frac{\floor{rn}}{n}\, .
\eee
Note that that by monotonicity of the trace norm under partial trace one has that
\begin{align*}
&\lim_{n\to\infty} \left\|\omega_n - \rho \right\|_1 \\ &\quad = \lim_{n\to\infty} \left\| \Tr_{A_1\ldots A_{i-1}A_{i+1}\ldots A_n} \left[ \sigma_n^{A_1\ldots A_n}\right] - \rho\right\|_1 \\
&\quad =  \lim_{n\to \infty} \left\| \Tr_{A_1\ldots A_{i-1}A_{i+1}\ldots A_n} \left[ \sigma_n^{A_1\ldots A_n} - \bigotimes\nolimits_{i=1}^n \rho^{A_i} \right] \right\|_1 \\
&\quad \leq \lim_{n\to\infty} \left\| \sigma_n^{A_1\ldots A_n} - \bigotimes\nolimits_{i=1}^n \rho^{A_i} \right\|_1 \\
&\quad = 0\, .
\end{align*}
Employing the lower semicontinuity of $C_f^\M$ (Proposition~\ref{properties CfM prop}(d)) yields
\bbb
C_f^\M(\rho) \leq \liminf_{n\to\infty} C_f^\M(\omega_n) \leq \liminf_{n\to\infty} \frac{\floor{rn}}{n} = r\, ,
\eee
which shows that an achievable rate $r$ cannot be larger than the uniform coherence of formation, completing the proof.
\end{proof}

\begin{rem}
The above proof strategy actually shows that in any resource theory \emph{all normalised, lower semicontinuous, superadditive monotones lower bound the dilution cost} on all states.
\end{rem}

\begin{rem}
As we have seen, our proof actually tells us more about PIO coherence dilution than what was stated in Theorem~\ref{C c PIO thm}. Namely, it follows from Lemma~\ref{PIO preserve co(M) lemma} that when $\rho \notin \co(\M)$ (and thus $C_f^\M(\rho) = \infty$) it is not possible to generate even a single copy of $\rho$ from an unlimited supply of uniformly coherent states with vanishing error.
\end{rem}

\begin{rem}
Theorem~\ref{C c PIO thm} combined with Proposition~\ref{properties CfM prop}(c) shows in particular that the coherence cost under PIO is fully additive. That is, an asymptotically optimal protocol to construct a composite state $\rho\otimes \sigma$ via PIO coherence dilution consists in creating $\rho$ and $\sigma$ separately. This in fact shows that \emph{all} coherence measures in Table~\ref{table operations/rates} are fully additive! Observe that statements of this kind are not known to hold in general in entanglement theory~\cite{Brandao2007}.
\end{rem}

Among the many consequences of Theorem~\ref{C c PIO thm}, one seems to us particularly surprising. Namely, one can show that there exists \emph{abyssally bound coherence} under PIO, that is, there are states with zero PIO distillable coherence yet \emph{infinite} PIO coherence cost. This particularly degenerate form of bound coherence is the signature of the extreme irreversibility of the resource theory of coherence under PIO.

\begin{cor}
Any qubit state as in Eq.~\eqref{qubit state} that satisfies $\min\{p,1-p\} < |z| < \sqrt{p(1-p)}$ is abyssally bound coherent under PIO, namely
\bb
C_{d,\PIO}(\rho) = 0\quad \text{and} \quad C_{c,\PIO}(\rho)=+\infty\, .
\ee
\end{cor}

\begin{proof}
Using Proposition~\ref{CfM qubits prop} it is immediate to see that such a state has infinite uniform coherence of formation and thus infinite PIO coherence cost by Theorem~\ref{C c PIO thm}. On the other hand, its quintessential coherence vanishes because it is a non-pure qubit state, making it PIO bound coherent by Theorem~\ref{solution thm}.
\end{proof}

For an explicit numerical example of an abyssally bound coherent state, take e.g.
\bb
\rho_0 \coloneqq \begin{pmatrix} 2/3 & 2/5 \\ 2/5 & 1/3 \end{pmatrix} .
\ee

\section{Discussion and conclusions} \label{sec conclusions}

We presented a general quantitative theory of coherence manipulation under strictly incoherent and physically incoherent operations in the asymptotic regime. We derived a simple analytical formula to compute the SIO/PIO distillable coherence on all finite-dimensional states in terms of the so-called quintessential coherence, thus extending the results of~\cite{bound-coherence}. Among other things, our construction shows that the optimal SIO distillation protocol can in fact be chosen within the much more restricted class of PIO. Since these operations are amenable to experimental implementations, our findings are likely to play a significant role in the near-term practice of quantum coherence manipulation. Our second result deals again with PIO, but in the somewhat complementary scenario of coherence dilution. We established a single-letter formula for the PIO coherence cost of all states: this is given by a convex-roof construction similar to that of the coherence of formation, which we dubbed uniform coherence of formation. A remarkable consequence of our analysis is that there is a set of nonzero volume entirely composed of states with finite PIO coherence cost. This can be interpreted by thinking of PIO as some intrinsically noisy operations; while coherence distillation requires noise subtraction and is thus often impossible, coherence dilution aims to produce (possibly) noisy states and can therefore become feasible. On the other hand, we have also uncovered the curious phenomenon of abyssally bound coherence under PIO, i.e.\ the existence of states with vanishing PIO distillable coherence that however have infinite PIO coherence cost.

In proving the above results we have introduced a number of novel techniques that may be of independent interest. First, to upper bound the SIO/IO distillation rates we constructed an entire family of new SIO monotones. In a \emph{tour de force} of linear algebra and probability theory that involves -- among other things -- Ger\v{s}gorin's circle theorem and a tweaked asymptotic equipartition property, we showed that their many properties make them powerful tools to investigate SIO. To analyse PIO coherence dilution we defined and studied the many properties of the uniform coherence of formation, most notably its superadditivity. To tackle the proof of the converse statement in absence of asymptotic continuity, we devised an alternative strategy that relying mainly on superadditivity and lower semicontinuity may carry over to other resource theories.

In conclusion, our findings complete the theoretical picture of asymptotic coherence manipulation under the classes of operations MIO/DIO/IO/SIO/PIO, solving some of the most pressing open problems. However, in the quest for a fully-fledged theory of quantum coherence manipulation the following outstanding questions seem important to us. (1) How do our results extend to other classes of incoherent operations, such as GIO and FIO considered in~\cite{deVicente2016}? (2) Is it possible to activate SIO/PIO bound coherent states by means of catalysts, as is the case for entanglement theory~\cite{Plenio-catalysis}? (3) What is the smallest physically meaningful set of operations that allows for coherence distillation~\cite{bound-coherence}?
(4) Can SIO simulate any quantum channel, provided that one is given a sufficiently large maximally coherent state to consume as a resource? Recently, it has been observed that the answer to this question is affirmative for the class IO~\cite{Chitambar-Hsieh-coherence, BenDana2017}. Our results instead imply that PIO simulation of certain channels is \emph{impossible}: by Lemma~\ref{PIO preserve co(M) lemma}, a channel that does not preserve the convex hull $\co(\M)$ of uniformly coherent states cannot be simulated via PIO, even if one allows for an arbitrarily large maximally coherent state to be consumed as a resource. In particular, this implies that the only unitaries that can be simulated in this way are trivially incoherent unitaries, in stark contrast with the case of IO~\cite{Chitambar-Hsieh-coherence}.
Finally, and related to this, we want to advertise as an outstanding question of mathematical interest (5) the membership problem for the set $\co (\M)$ of states with a finite PIO coherence cost, beyond the qubit case. We do not yet know, for instance, whether it can be decided efficiently in general by checking some finite set of conditions, as is the case for qubit states (Proposition~\ref{CfM qubits prop}).

\appendices

\section{PIO coherence cost of uniformly coherent states}

\begin{lemma} \label{PIO cost pure states lemma}
The PIO coherence cost of a uniformly coherent state $\Psi_k\in \M_k$ is no larger than
\bb
C_{c,\PIO}(\Psi_k) \leq \log k = C_f^\M(\Psi_k)\, .
\label{PIO cost pure states}
\ee
\end{lemma}

\begin{rem}
It is easy to show that the l.h.s.\ and r.h.s.\ of~\eqref{PIO cost pure states} in fact coincide. This follows e.g.\ from the fact that the IO coherence cost of $\Psi_k$ is exactly $\log k$~\cite{Xiao2015}, other than from our Theorem~\ref{C c PIO thm}.
\end{rem}

\begin{proof}[Proof of Lemma~\ref{PIO cost pure states lemma}]
If $\log k$ is an integer there is nothing to prove, because by relabelling the basis vectors one can transform $n \log k$ copies of $\Psi_2$ into exactly $n$ copies of $\Psi_k$ with no error. Since binary logarithms of integers are either integer or irrational, we can henceforth assume that $\log k$ be irrational. For some fixed $\delta,\epsilon>0$ and sufficiently large $n$, we proceed to show that it is possible to convert $\floor{n (\log k + \delta)}$ independent copies of the coherence bit $\Psi_2$ into $n$ copies of $\Psi_k$ with an error at most $\epsilon$.

Using the fact that the sequence $\left( \{nx\}\right)_{n\in \N}$ of fractional parts of the integer multiples of a fixed irrational number forms a dense subset of $[0,1)$, it is not too difficult to show that eventually in $n$ one can pick integers $M,N$ such that
\bb
n \log k \leq M + \log \left(1-\frac{\epsilon}{2}\right) \leq N \log k \leq M \leq n (\log k + \delta) .
\label{integer inequality}
\ee
Clearly, one has $M\coloneqq \ceil{N \log k}$; moreover, it also holds that $N\geq n$. Now, up to discarding some subsystems we can assume that our initial state is of the form $\Psi_2^{\otimes M}$.  Let us decompose the corresponding Hilbert space $\left(\C^2\right)^{\otimes M}$ as a direct sum of the subspace $H_0$ spanned by the first $k^N\leq 2^M$ vectors of the computational basis and its orthogonal complement $H_1$. Call $\Pi_0$ and $\Pi_1$ the projectors onto those subspaces. Observe that
\bbb
\Tr\left[ \Pi_0 \, \Psi_2^{\otimes M} \right] = \frac{k^N}{2^M} \geq 1-\frac{\epsilon}{2}\, ;
\eee
moreover, $\Pi_0 \ket{\Psi_2}^{\otimes M}$ is equivalent to $\ket{\Psi_k}^N$ up to relabelling of the basis vectors. Since the probability of the PIO measurement $\{\Pi_0, \Pi_1\}$ yielding the outcome $0$ is at least $1-\epsilon/2$, performing said measurement and outputting a junk state in case it does not succeed produces the state $\Psi_k^{\otimes N}$ with error at most $\epsilon$ as measured by the trace norm. Discarding some output states we finally arrive at $\Psi_k^{\otimes n}$, as claimed.
\end{proof}

\section*{Acknowledgements}
I thank Gerardo Adesso, Guillaume Aubrun, Benjamin Morris, and especially Martin Plenio, Bartosz Regula and Andreas Winter for many insightful discussions on the topic of coherence. I am also grateful to Violetta Val\'ery for inspiration. Finally, I acknowledge financial support from the European Research Council (ERC) under the Starting Grant GQCOP (Grant No.~637352).


\bibliography{biblio}

%



\newpage
\begin{IEEEbiographynophoto}{Ludovico Lami} received a M.Sc.\ degree in Physics from the Universit\`a di Pisa, Pisa, Italy, in 2014, a Diploma in Physics at the Scuola Normale Superiore, Pisa, Italy, in 2015, and a Ph.D.\ degree from the Departament de F\'isica of the Universitat Aut\`onoma de Barcelona, Barcelona, Spain, in 2017. He is currently a Research Fellow at the University of Nottingham. His research interests lie in quantum information, also with continuous variable, and in foundational aspects of quantum physics.
\end{IEEEbiographynophoto}




\end{document}

%% file: IEEE-preambole.tex
\ifCLASSINFOpdf
\else
\fi

%% file: def.tex
\usepackage[latin1]{inputenc}
\usepackage{amsthm}
\usepackage{amssymb}
\usepackage{amsmath}
\usepackage{bbold}
\usepackage{bbm}
\usepackage{nonfloat}
\usepackage[pdftex]{hyperref}
\usepackage{braket}
\usepackage{dsfont}
\usepackage{mathdots}
\usepackage{mathtools}
\usepackage{enumerate}
\usepackage{csquotes}
\usepackage{stmaryrd}
\usepackage[cal=boondox]{mathalfa}
\usepackage{graphicx}
\usepackage{stackengine}
\usepackage{scalerel}
\usepackage[table]{xcolor}
\usepackage{array}
\usepackage{makecell}
\newcolumntype{x}[1]{>{\centering\arraybackslash}p{#1}}
\usepackage{tikz}
\usepackage{booktabs}

\newcommand{\ra}[1]{\renewcommand{\arraystretch}{#1}}

\newtheorem{thm}{Theorem}
\newtheorem*{thm*}{Theorem}
\makeatletter
\newcommand{\setthmtag}[1]{
  \let\oldthethm\thethm
  \renewcommand{\thethm}{#1}
  \g@addto@macro\endthm{
    \addtocounter{thm}{-1}
    \global\let\thethm\oldthethm}
  }
\makeatother

\newtheorem{prop}[thm]{Proposition}
\newtheorem*{prop*}{Proposition}
\newtheorem{lemma}[thm]{Lemma}
\newtheorem*{lemma*}{Lemma}
\newtheorem{cor}[thm]{Corollary}
\newtheorem*{cor*}{Corollary}

\newtheorem*{cj*}{Conjecture}
\newtheorem{Def}[thm]{Definition}
\newtheorem*{Def*}{Definition}

\theoremstyle{definition}
\newtheorem{rem}{Remark}
\newtheorem*{note}{Note}

\newcommand{\bb}{\begin{equation}}
\newcommand{\bbb}{\begin{equation*}}
\newcommand{\ee}{\end{equation}}
\newcommand{\eee}{\end{equation*}}
\newcommand*{\coloneqq}{\mathrel{\vcenter{\baselineskip0.5ex \lineskiplimit0pt \hbox{\scriptsize.}\hbox{\scriptsize.}}} =}
\newcommand*{\eqqcolon}{= \mathrel{\vcenter{\baselineskip0.5ex \lineskiplimit0pt \hbox{\scriptsize.}\hbox{\scriptsize.}}}}
\newcommand\floor[1]{\left\lfloor#1\right\rfloor}
\newcommand\ceil[1]{\left\lceil#1\right\rceil}
\newcommand{\texteq}[1]{\stackrel{\mathclap{\scriptsize \mbox{#1}}}{=}}
\newcommand{\textleq}[1]{\stackrel{\mathclap{\scriptsize \mbox{#1}}}{\leq}}

\newcommand{\textgeq}[1]{\stackrel{\mathclap{\scriptsize \mbox{#1}}}{\geq}}
\newcommand{\ketbra}[1]{\ket{#1}\!\!\bra{#1}}
\newcommand{\ketbraa}[2]{\ket{#1}\!\!\bra{#2}}
\newcommand{\sumno}{\sum\nolimits}

\newcommand{\id}{\mathds{1}}
\newcommand{\R}{\mathds{R}}
\newcommand{\N}{\mathds{N}}
\newcommand{\C}{\mathds{C}}

\newcommand{\PIO}{\mathrm{PIO}}
\newcommand{\SIO}{\mathrm{SIO}}
\newcommand{\IO}{\mathrm{IO}}
\newcommand{\DIO}{\mathrm{DIO}}
\newcommand{\MIO}{\mathrm{MIO}}

\DeclareMathOperator{\Tr}{Tr}
\DeclareMathOperator{\rk}{rk}

\DeclareMathOperator{\co}{conv}

\DeclareMathOperator{\Span}{span}
\DeclareMathAlphabet{\pazocal}{OMS}{zplm}{m}{n}

\DeclareMathOperator{\pr}{Pr}
\DeclareMathOperator{\supp}{supp}

\DeclareMathOperator{\Id}{id}

\newcommand{\lsmatrix}{\left(\begin{smallmatrix}}
\newcommand{\rsmatrix}{\end{smallmatrix}\right)}

\stackMath

\stackMath

\makeatletter
\newcommand*\rel@kern[1]{\kern#1\dimexpr\macc@kerna}
\newcommand*\widebar[1]{%
  \begingroup
  \def\mathaccent##1##2{%
    \rel@kern{0.8}%
    \overline{\rel@kern{-0.8}\macc@nucleus\rel@kern{0.2}}%
    \rel@kern{-0.2}%
  }%
  \macc@depth\@ne
  \let\math@bgroup\@empty \let\math@egroup\macc@set@skewchar
  \mathsurround\z@ \frozen@everymath{\mathgroup\macc@group\relax}%
  \macc@set@skewchar\relax
  \let\mathaccentV\macc@nested@a
  \macc@nested@a\relax111{#1}%
  \endgroup
}